\documentclass[12pt, amsfonts, epsfig, reqno, a4paper]{amsart}
\usepackage[hyperindex,breaklinks]{hyperref}
\usepackage{amsmath}
\usepackage{amsfonts}
\usepackage{mathrsfs}
\usepackage{graphicx}
\usepackage{color}
\usepackage{slashed}
\usepackage{braket}
\usepackage[margin=1in]{geometry}
\newcommand{\PP}{{\mathbb{P}}}
\newcommand{\C}{{\mathbb{C}}}
\newcommand{\bT}{{\mathbf{T}}}
\newcommand{\E}{{\mathcal{E}}}
\newcommand{\Z}{{\mathbb{Z}}}
\newcommand{\CR}{{\mathrm{CR}}}
\newcommand{\M}{\overline{\mathcal{M}}}
\newcommand{\bp}{{\mathbf p}}

\newcommand{\ev}{{\mathrm{ev}}} 
\newcommand{\half}{\frac12}

\newcommand{\si}{S_\mathrm{inf}}
\newcommand{\Id}{\mathrm{Id}}
\newcommand{\Sc}{\mathcal{S}}
\newcommand{\un}{\mathrm{un}}
\newcommand{\diag}{\mathrm{diag}}
\newcommand{\ndiag}{\mathrm{ndiag}}

\newtheorem*{theorem2.1}{Theorem 2.1}
\newtheorem*{theorem2.1'}{Theorem 2.1'}
\newtheorem*{theorem2.2'}{Theorem 2.2'}
\newtheorem*{theorem2.2}{Theorem 2.2}
\newtheorem*{proposition2.2}{Proposition 2.2}
\newtheorem*{lemma2.1}{Lemma 2.1}
\newtheorem*{theorem3.1}{Theorem 3.1}
\newtheorem{theorem}{Theorem}[section]
\newtheorem{corollary}{Corollary}[section]
\newtheorem{lemma}{Lemma}[section]
\newtheorem{remark}{Remark}[section]
\newtheorem{proposition}{Proposition}[section]
\newtheorem{definition}{Definition}[section]

\author{Chongyao Chen }
\email{chenchongyao@pku.edu.cn}
\address{Chongyao, Chen, Department of Physics and State
Key Laboratory of Nuclear Physics and Technology, Peking
University, Beijing}

\author{Shuai Guo}
\email{guoshuai@math.pku.edu.cn}
\address{Shuai, Guo, School of Mathematical Sciences, Peking University, Beijing}

\title[{Quantum curve and  bilinear Fermionic form for {$\mathbb{P}$\lowercase{$[r]$}}}]{Quantum curve and  bilinear Fermionic form for  the orbifold Gromov-Witten theory of $\mathbb{P}[r]$}

\begin{document}
\maketitle
%%%%%%%%%%%%%%%%%%%%%%%%%%%%%%%%%%%%%%%%%%%%%%%%%%%%%%%%%%%%%%%%%%%%%%%%%%%%%%%%%%%

\begin{abstract}
	We construct the quantum curve for the Baker-Akhiezer function of the orbifold Gromov-Witten theory of the weighted projective line $\mathbb P[r]$. Furthermore, we deduce the explicit bilinear Fermionic formula for the (stationary) Gromov-Witten potential via the lifting operator contructed from the Baker-Akhiezer function. 
	%A generalization of our method for deducing bilinear Fermionic form to other known cases will appear in a sequel.{}
\end{abstract}

\setcounter{section}{-1}
\setcounter{tocdepth}{1}
\tableofcontents

\section{Introduction}
%\section{Main results} 
\label{s0}

It is a general phenomenon that the generating functions of various kinds of enumerative problems are the tau-functions of certain integrable hierarchies, after suitable change of variables. Furthermore, when the integrable hierarchies are reductions of ($n\/$-component) KP hierarchies, it is conjectured that the Baker-Akhiezer  functions $\Psi(t,x)$ multiplied by a factor of ``unstable contribution" is the solution of a differential equation of the form
$$
\hat H_t(x,-\hbar \partial_x) e^{S_{\un}(t,x)/\hbar} \Psi(t,x) =0 ,
$$
such that the operator $\hat H_t(x,\hbar \partial_x)$ is a quantization of some function $H_t(x,p)$ over a two dimensional affine space $\mathbb A$, which defines a family of spectral curves
$$
C_t = \{ H_t(x,p)\in \mathbb A:  H(t,x)=0 \} .
$$
The operator $\hat H_t$ is called \emph{quantum curve} in the literature \cite{ADKMV,dijkgraaf2007two,dijkgraaf2008supersymmetric,dijkgraaf2009quantum,gukov2012polynomial}.

The quantum curves have been thoroughly studied for many cases in the past few years, including: various types of Hurwitz numbers\cite{mulase2012spectral,zhou2012quantum,mulase2013spectral,liu2013quantum,bouchard2014mirror,alexandrov2016ramifications}, Kontsevich-Witten and r-spin intersection numbers\cite{zhou2012intersection,ding2016r}, etc.
Yet for Gromov-Witten theories only limited results have been reached mathematically\cite{dunin2017quantum,norbury2016quantum}. Also, most of the existing examples of quantum curves have underlying rational classical curves.  Therefore, we hope our aim will be shading some lights in these directions. \par

\medskip

For $r\in \mathbb Z_+$, we denote by $\mathbb P[r]$  the weighted projective line which has a single stack point of order $r$ at $\infty$.
Let  $\M_{g,m,\gamma}(\mathbb P[r],d)$ be the moduli space of stable maps to $\mathbb P[r]$ of degree $d$, where $\gamma=(\gamma_1,\cdots,\gamma_n)$ with $\gamma_i \in \mathbb Z_r$ gives the $n$-tuple of monodromies.

We introduce the genus $g$, $n$-point stationary correlators
\begin{equation}\label{gwinvariant}
\left<h \bar{\psi}_1^{k_1},\cdots, h \bar{\psi}_n^{k_n} \right>_{g,m+n}^{\mathbb P[r]} :=
\sum_{d\geq 0} q^d \int_{[\M_{g,n,\emptyset}(\mathbb P[r],d)]^{Vir}} \prod_{i=1}^n (\ev_i^* h) \bar{\psi}_i^{k_i} .
\end{equation}
where $h$  is the hyperplane class  of $\mathbb P[r]$, and $\bar{\psi}_i$ the first Chern class of the cotangent line
bundle $L_i$ on the moduli space of stable maps.\footnote{ Here we use  $\bar{\psi}_k$ to denote the psi-classes,  since we will denote by $\psi_k$  the free Fermions in this paper.}
Let 
$$\bT = \bT(z) = \sum_{k\geq 0} T_k \, h\cdot z^k\in H^2_{\CR}({\mathbb P[r]} )[[z]],$$
we define the stationary
Gromov-Witten potential function by
$$
F(\hbar,\bT) := \sum_{g,n\geq 0} \frac{\hbar^{g-1}}{n!} \left< \bT(\bar{\psi}_1),\cdots \bT(\bar{\psi}_n) \right>_{g,n}^{\mathbb P[r]} ,
$$
and we define the Gromov-Witten partition function by
\begin{equation}\label{gwp}
Z(\hbar,\bT)= \exp F(\hbar,\bT).
\end{equation}

\subsection{From Gromov-Witten potential to integrable hierarchy and quantum curve} \label{secquantumcurve}

The first main result of this paper, is to prove that the quantum curve for the Gromov-Witten theory of the weighted projective line $\mathbb P[r]$ is given by 
\begin{equation} \label{QC}
\hat H_t(x, -\hbar\partial_x) = e^{ -\hbar \partial_x} + q^r e^{r\,t\,\hbar}\, e^{r \hbar \partial_x} -x   + \textstyle \frac{1}{2}\hbar  .
\end{equation}

We now give the precise statement. 
Let $\Psi(t,x)$ be the Baker-Akhiezer  function defined from the Gromov-Witten partition function \eqref{gwp} via the following change of variables: \footnote{
	See also \eqref{BAfunc} in Section~\ref{bafunction} for the standard definition via the Fermionic Fields.
	The  $t=0$ case of
	the specialization above, is also called principal specialization in the literature.}
\begin{equation}\label{kpprin}
%\begin{split}
\Psi(t,x) = e^{tx}\cdot  \frac{Z(\hbar,\bT)|_{T_0=t-x^{-1}, \ T_k = -k!x^{-k-1} \ \forall k>0}}{Z(\hbar,\bT)|_{T _0=t , \ T_k = 0 \ \forall k>0}}
%\Psi^*(t,y) &= e^{-ty}\cdot  \frac{Z(\hbar,\bT)|_{T_0=-t+y^{-1}, \ T_k = k!y^{-k-1} \ \forall k>0}}{Z(\hbar,\bT)|_{t _0=t , \ T_k = 0 \ \forall k>0}}.
%\end{split}
\end{equation}\par

%Similarly, one can also define the adjoint Baker-Akhiezer function

%$$
%\Psi(t,y) = e^{-ty}\cdot  \frac{Z(\hbar,\bT)|_{T_0=-t+y^{-1}, \ T_k = k!y^{-k-1} \ \forall k>0}}{Z(\hbar,\bT)|_{t _0=t , \ T_k = 0 \ \forall k>0}}.
%$$

%\begin{remark}
%Note that by equation \eqref{equibasis}, picking $t^0_k=0$ $\forall k\geq 0$ means we set
%$$
%x^*_k = t^1_{k-1},\qquad x_k =0
%$$
%\end{remark}
%%%%%%%%%%%%%%%%%%%%%%%%%%%%%%%%%%%%%%%%%%%%%%%%%%%%%%%%%%%%%%%%%%%%%%%%%%%%%%%%%%%%%%%%%%%%%%%%%%%%%%%%%%%%%%%%%%%%%%%%%%%%%%%%%%%%%%%%%%%%%%%%%%%%%

\begin{theorem2.2}
	We define  the wave function of the Gromov-Witten theory of $\mathbb{P}[r]$ by
	\begin{equation}
	\Phi(t,x) = e^{S_\un(t,x)/\hbar}\Psi(t,x)
	\end{equation}
where $\Psi(t,x)$ is the Baker-Akhiezer function \eqref{kpprin} and the unstable contribution is  defined as
$$
S_\un(t,x) = -x\ln x+x-t\hbar x.
$$
	then the wave function $\Phi(t,x)$ satisfyies the following ``quantum curve" equation
	$$
	\hat{H}_t(x,-\hbar\partial_x) \, \Phi(t,x) = 0.
	$$
	where the quantum curve $\hat{H}_t(x,-\hbar\partial_x)$ is the differential operator  defined in \eqref{QC}.
	%\begin{equation*}
	%\begin{split}
	%&\hat{H}(x,-\hbar\partial_x) = e^{-\hbar\frac{\partial}{\partial x}}+q^re^{tr}e^{r\hbar\frac{\partial}{\partial x}}-x+\frac{\hbar}{2}\\
	%&\hat{H}^*(y,\hbar\partial_y) = e^{\hbar\frac{\partial}{\partial y}}+q^re^{tr}e^{-r\hbar\frac{\partial}{\partial x}}-y-\frac{\hbar}{2}.
	%\end{split}
	%\end{equation*}
\end{theorem2.2}\par
\begin{remark}
This quantum curve can be regarded as a quantization of the classical spectral curve
$$
\{(p,x)\in \C \times \C : \  e^{ p} + q^r \, e^{-rp} -x  =0\} .
$$
In particular, our result specialize to those obtained in \cite{dunin2017quantum} for stationary $\PP^1$ ($r=1$).
\end{remark}

The above theorem is   a direct consequence of  the following explicit formula%, by wave function we refer to $\Phi(t,x) = e^{S_\un(t,x)/\hbar}\Psi(t,x) $
\begin{theorem2.1}
	The wave function for the orbifold Gromov-Witten theory of $\mathbb{P}[r]$ has the following closed form
	\begin{equation*}
	\Phi(t,x) = e^{S_\un(t,x)/\hbar}\Psi(t,x) \sim\left(2\pi\right)^{\half} \sum_{d=0}^{\infty}\frac{(-1)^dq^{rd}e^{t\hbar rd}}{r^dd!}\frac{\hbar^{-\frac{x}{\hbar}-(r+1)d}}{\Gamma(\frac{x}{\hbar}+rd+\frac{1}{2})},
	\end{equation*}
	%The counterpart for the adjoint wave function is
	%\begin{equation*}
	%\Phi^*(t,y) = e^{S^*_\un(t,y)/\hbar}\Psi^*(t,y) \sim \left(2\pi\right)^{\half}\sum_{d=0}^{\infty}\frac{(-1)^dq^{rd}e^{t\hbar rd}}{r^dd!}\frac{\hbar^{\frac{y}{\hbar}-rd+\frac{1}{2}}}{\Gamma(-\frac{y}{\hbar}+rd-\frac{1}{2})},
	%\end{equation*}
	%with the unstable contribution
	%$$
	%S^*_\un(t,y) = y\ln (-y)-y+t\cdot y,
	%$$
	%where $\sim$ refers to the asymptotic expansion at $x,y\rightarrow\infty,|\arg(x)|<\pi$ and $|\arg(y)|>\pi$,resp.
	where $\sim$ refers to the asymptotic expansion at $x\rightarrow\infty,|\arg(x)|<\pi$.
\end{theorem2.1}\par

\subsection{The canonical bilinear Fermionic form} \label{sectwopointfunction}

To connect the Gromov-Witten partition function with the KP hierarchy, we need to perform the following change of variables
$
T _k \mapsto k! p_{k+1}
$. 
We denote the resulting generating function by
\begin{equation}
\tau(\hbar,\bp):=Z(\hbar,\bT) |_{T_k \mapsto k! p_{k+1}}.
\end{equation}
By the result of \cite{okounkov2006equivariant,johnson2009equivariant}, $\tau(\hbar,\bp)$ is a $\tau$-function of the KP hierarchy. In the literature, the $\tau$-functions are also considered as functions of  KP times variables $t_k:=\frac{p_k}{k}$.

\smallskip

By Sato's theory, any $\tau$-function of the KP hierarchy  is specified by a point in the semi-infinite Grassmannian of $V$, which can be Pl{\"u}cker embedded into  the projectivization of $\Lambda_0^{\frac{\infty}{2}}V$ as a cone. Such a point can be considered as a transformation from the vacuum $|0 \rangle$ by an element $G\in \widehat{GL}(\infty)$. Although $G$ does not correspond to an unique element in $\widehat{\mathfrak{gl}}(\infty)$ via exponential map, there is a canonical way to put it, i.e. \emph{the canonical bilinear Fermionic form}, which gives a canonical basis of $V$ (c.f. Section~\ref{canonicalbasis}, see also \cite{kharchev1998kadomtsev,ADKMV,alexandrov2015enumerative}).\par

\smallskip

Now we construct the canonical bilinear Fermionic form:
An element $G\in \widehat{GL}(\infty)$ is called a \emph{Bogoliubov transformation} if it is of the form
\begin{equation}\label{uniform}
G =  e^{\sum_{i,k=0}^{\infty}b_{i,k}\psi^*_{-k-\half}\psi_{i+\half}}  .
\end{equation}
As we have mentioned,  any solution of the KP hierarchy  can be considered as an element $|V \rangle$ in  $\Lambda_0^{\frac{\infty}{2}}V$
(see Section~\ref{s1.1} for the precise definition). Then the canonical bilinear Fermionic form of $|V\rangle$ is given by the following Bogoliubov transformation
 of the vacuum:
 $$
|V\rangle = G |0 \rangle
 $$
where $b_{i,k}$ are determined uniquely by the following Fermionic two point function
$$
B(x,y) :=\sum_{i,j \geq 0}b_{i,j} x^{-i-1} y^{-j-1}  = \frac{\braket{0|\psi^*(y)\psi(x)|V}}{\braket{0|V}}.
$$
%We will call \eqref{uniform} the canonical bilinear Fermionic form.

The second main result gives the closed form of the  canonical Fermionic bilinear  form:

\begin{theorem3.1}
	The tau function of Gromov-Witten theory of $\mathbb P[r]$ can be written as the following explicit Fermionic bilinear \footnote {{ Here we have identified the space  $\mathbb C[[\hbar,q]][\mathbf p]$  with the Fermionic fock space by using the Boson-Fermion correspondence, see for example Equation~\eqref{bfcorrespondence} for the explicit construction. }}
	\begin{equation}
	\tau(\hbar,\mathbf p) = e^{\sum_{i,j \geq 0} b_{i,j}\psi^*_{-k-\half}\psi_{i+\half}} |0\rangle,
	\end{equation}
	with the explicit form of the generating function for $b_{i,j}$ defined by
\begin{equation} \label{Bxy}
B(x,y)= \frac{e^{\frac{x}{\hbar}\ln \frac{x}{\hbar}-\frac{x}{\hbar}}}{e^{\frac{y}{\hbar}\ln \frac{y}{\hbar}-\frac{y}{\hbar}}}\sum_{d=1}^{\infty}\frac{q^{rd}}{d\cdot r^d\hbar^{(r+1)d}}\sum_{k=0}^{d-1}\frac{(-1)^{k-1}}{k!(d-1-k)!}\sum_{n=1}^{r}
\frac{{\Gamma(\frac{y}{\hbar}+n-r(d-k)-\frac{1}{2})}}{{\Gamma(\frac{x}{\hbar}+rk+n+\frac{1}{2})}}
%x_{[rk+n]}y_{[r(d-k)+1-n]}.
\end{equation}
where we expand the RHS as an asymptotic series in $ Q[[\hbar x^{-1},\hbar y^{-1}]]$.
% where $x_{[k]}: = \prod_{i=1}^{k}\frac{\hbar}{x+(i-\half)\hbar}, k>0$, $x_{[k]} = \prod_{i=0}^{-k-1}\frac{x-(i+\half)\hbar}{\hbar}, k<0$ and $y_{[k]}  = \prod_{i=1}^{k}\frac{\hbar}{y-(i-\half)\hbar},  k>0$, $y_{[k]} = \prod_{i=0}^{-k-1}\frac{y+(i+\half)\hbar}{\hbar},  k<0$. We also make the convention $x_{[0]}=y_{[0]}:=1$. 

\end{theorem3.1} 

We will prove this theorem directly from Theorem \ref{onept}. Hence, we give an algorithm which deduces the explicit formula for the all genus (stationary) Gromov-Witten potential of $\mathbb P[r]$ directly via the Baker-Akhiezer function.

\begin{remark}
In this paper we will always working in and fixed value of $r$, all quantity that varies along $r$ will be considered in particular $\PP[r]$, especially for $r=1$, we refer to the stationary Gromov-Witten theory of $\PP^1$.
\end{remark}

The paper is organized as follows. % In Section \ref{s0}, we give the precise statements of our main results. 
In Section \ref{s1}, we review the preliminaries in infinite wedge and orbifold Gromov-Witten theory concerning our work. In Section \ref{s2}, we review the operator formalism of orbifold Gromov-Witten theory of $\PP[r]$, derive the closed form of wave-function and prove it satisfy the quantum curve equation. In Section \ref{s3} we will derive the canonical Fermionic bilinear.
Some technical proofs will be given in the Appendix.\par

Finally, we want to remark that the idea to dedue the Fermionic bilinear form via the lifting operators and  the Baker-Akhiezer  functions in Section~\ref{s3}   works for more general cases.
%  Namely, when  considering Baker-Akhiezer  functions as the family of the quantum curves, one can reconstruct the  tau function of the KP hierarchy just by the family of the  equations.  
We will address  to other cases in a sequel.

%%%%%%%%%%%%%%%%%%%%%%%%%%%%%%%%%%%%%%%%%%%%%%%%%%%%%%%%%%%%%%%%%%%%%%%%%%%%%%%%%%%%%%%%%%%%%%%%%%%%%%%%%%%%%%%%%%%%%%%%%%%%%%%%%%%%
\section{Infinite Wedge and Orbifold Gromov Witten theory}\label{s1}
%{\color{red}(In this section I will list the theorems, propositions and equations needed in the remaining sections, due to the needs of later reference, so this section may seems not consecutive. Whenever there is a matter of choice about the convention that we have discussed, the convention I chose here will coincide with the results in later sections.)}\par
%%
%%
In this section, we defined the infinite wedge space $\Lambda^{\frac{\infty}{2}}V$ as the Fermionic Fock space, and discussed the operators on $\Lambda_0^{\frac{\infty}{2}}V$. After the introduction of necessary ingredients, we will see the generating series of stationary orbifold Gromov-Witten theory of $\PP[r]$ can be written as a tau-function of KP integrable hierarchy.\par

\subsection{The Fermionic Fock space}\label{s1.1}
The Fermionic Fock space provides the arena for the operator formalism, which we now briefly review, for more details, we refer to \cite{okounkov2001infinite,okounkov2006equivariant,alexandrov2015enumerative}. Let $V$ be the operator spanned by $\{\underline{k}\}, k\in\mathbb{Z}+\half$, i.e.
$$
V = \bigoplus_{k\in\mathbb{Z}+\half}\C\underline{k}.
$$

The Fermionic Fock space $\Lambda^{\frac{\infty}{2}}V$ is defined as the semi-infinite wedge space of $V$. More precisely, let $\mathcal{C}$ be the collection of all ordered subsets $S = \{s_1,s_2,\cdots\}$ of $\mathbb{Z}+\half$, satisfying
\begin{equation*}
\begin{split}
    &\mathrm{(1)}\,\#\{s_i\geq0\}<\infty,\\
    &\mathrm{(2)}\,\#\mathbb{Z}_{-}\setminus\{s_i<0\}<\infty,
\end{split}
\end{equation*}
and we denote by $v_{S}$  the following wedge product
$$
v_{S} = \underline{s_1}\wedge\underline{s_2}\wedge\underline{s_3}\wedge\cdots.
$$
Then we have
$$
\Lambda^{\frac{\infty}{2}} V = \bigoplus_{S\in\mathcal{C}}\C v_{S}.
$$

On $\Lambda^{\frac{\infty}{2}}V$, one can define an Hermitian inner product $\left<\cdot,\cdot\right>$, where $\{v_{S}\}$ consists an orthonormal basis.

\begin{remark}
In our notation, the indices for the basis of $V$ are half integer, which can be understood as shifted by $-\frac{1}{2}$ from the integer notation $\underline{x^{k}}$, i.e. we have for $k\in\mathbb{Z}$, $\underline{k-\half} = \underline{x^k}$. These two notation will be used interchangably.
\end{remark}

The Fermionic creator $\psi_k$ acting on $\Lambda^{\frac{\infty}{2}} V$ as
%{\Chen: Which notation should we use in the $\psi_k$s? k is half integer or integer? From your notation in Eq. (2) it seems to be half integer, which I will use in the sequel.}
$$
\psi_k\cdot v = \underline{k}\wedge v,\quad k\in\Z+\half.
$$
meanwhile, the Fermionic annihilator $\psi_k^*$ is defined as the adjoint of $\psi_k$ w.r.t. the inner product $\left<\cdot,\cdot\right>$.\par

The  creators and annihilators satisfy the anti-canonical commutation relations:
\begin{equation*}
\begin{split}
\{\psi_i,\psi_j^*\} :=& \psi_i\psi_j^*+\psi_j^*\psi_i = \delta_{i,j}\\
\{\psi_i,\psi_j\} &= \{\psi^*_i,\psi_j^*\}=0,
\end{split}
\end{equation*}
and the normal-ordering defined as
\begin{equation*}
:\psi_i\psi_j^*:=\left\{
\begin{aligned}
&\psi_i\psi_j^*,&\quad j>0\\
&-\psi_j^*\psi_i,&\quad j<0\\
\end{aligned}\right.
\end{equation*}.

Denote, $E_{i,j}, i,j,\in\Z+\half$ as the single entry matrices in $\mathfrak{gl}(\infty)$, which form a standard basis of it. The central extension can be manifested via
$$
\hat{r}(E_{i,j}) = :\psi_i\psi^{*}_j:,
$$
which expands linearly on $\mathfrak{gl}(\infty)$ and forms a projective representation $\hat{r}$ of $\mathfrak{gl}(\infty)$ on $\Lambda^{\frac{\infty}{2}}V$. $\hat{r}$ can also being regarded as a linear representation of $\widehat{\mathfrak{gl}}(\infty): = \mathfrak{gl}(\infty)\oplus c\C$, i.e. the central extension of $\mathfrak{gl}(\infty)$.\par

The following two elements in $\mathfrak{gl}(\infty)$ serve great importance in the sequel, we define the charge operator $C$ and energy operator $H$ as

$$
C = \sum_{k\in\Z+\half}E_{kk},\quad H = \sum_{k\in\Z+\half}kE_{kk}.
$$

Clearly, $v_S,\forall S\in\mathcal{C}$ is a eigenvector for $\hat{r}(C)$ and $\hat{r}(H)$, thus $\Lambda^{\frac{\infty}{2}}V$ is bigraded under $\hat{r}(C)$ and $\hat{r}(H)$. Under the grading of $\hat{r}(C)$, we have
$$
\Lambda^{\frac{\infty}{2}}V := \bigoplus_{k}\Lambda_k^{\frac{\infty}{2}}V.
$$
where $\Lambda_k^{\frac{\infty}{2}}V$ consists vectors $v_S$ with charge $k$. It is clear that each element in $\Lambda_0^{\frac{\infty}{2}}V$ can be identified with a finite partition $\lambda$ via

$$
v_\lambda = \underline{x^{\lambda_1-1}}\wedge\underline{x^{\lambda_2-2}}\wedge\underline{x^{\lambda_3-3}}\wedge\cdots,
$$
whose eigenvalue under $\hat{r}(H)$ is
$$
\hat{r}(H) v_\lambda = |\lambda| v_\lambda.
$$

We will denote the vacuum, i.e. the 0 energy state in  $\Lambda_0^{\frac{\infty}{2}}V$ corresponding to the empty Young diagram as $v_{\emptyset}=| 0\rangle$.

%%%%%%%%%%%%%%%%%%%%%%%%%%%%%%%%%%%%%%%%%%%%%%%%%%%%%%%%%%%%%%%%%%%%%%%%%%%%%%%%%%%%%%%%%%%%%%%%%%%%%%%%%%%%%%%%%%%%%%%%%%%%%%%%%%%%

\subsection{Representations of infinite dimensional Lie group and Lie algebra}
In this subsection, we will carefully investigate the infinite dimensional Lie group $GL(\infty)$, Lie algebra $\mathfrak{gl}(\infty)$, their central extensions and their (projective) representations, one need to be careful about the distinguish between linear and projective representations. we refer to \cite{kac2013bombay} for more details.\par

\begin{remark}
Same as those in \cite{okounkov2006equivariant,johnson2009equivariant}, in this paper we will write the matrix with integer indices in descending order from the upper left to the lower right. If one want to use the ascending order, then the action on $\Lambda_0^{\frac{\infty}{2}}V$ will become right multiplication instead of left. 
\end{remark}

The fundamental (linear) representation for Lie group $GL(\infty)$ and Lie algebra $\mathfrak{gl}(\infty)$ are naturally defined by matrix multiplication on $V$, which in turn induced the following linear representations on $\Lambda_{0}^{\frac{\infty}{2}}V$. For $a\in\mathfrak{gl}(\infty),\,A\in GL(\infty)$

\begin{equation*}
\begin{split}
&r(a)\cdot (\underline{x^{i_1}}\wedge \underline{x^{i_2}}\wedge \underline{x^{i_3}}\wedge\cdots) = (a\cdot \underline{x^{i_1}})\wedge \underline{x^{i_2}}\wedge \underline{x^{i_3}}\wedge\cdots+\underline{x^{i_1}}\wedge (a\cdot \underline{x^{i_2}})\wedge \underline{x^{i_3}}\wedge\cdots+\cdots\\
&R(A)\cdot (\underline{x^{i_1}}\wedge \underline{x^{i_2}}\wedge \underline{x^{i_3}}\wedge\cdots) = A\cdot \underline{x^{i_1}}\wedge A\cdot \underline{x^{i_2}}\wedge A\cdot \underline{x^{i_3}}\wedge\cdots,
\end{split}
\end{equation*}
where the dot in rhs are matrix multiplications, and $R$ is unitary. As we mentioned in Section \ref{s1.1}, $\hat{r}$ defines a projective representation of $\mathfrak{gl}(\infty)$ via
$$
\hat{r}(E_{i,j}) = :\psi_i\psi^{*}_j:,
$$
The projective representation $\hat{r}$ differs from the linear one $r$ by an a priori chosen 2-cocyle $c(\cdot,\cdot)$ in the Lie algebra cohomology of $\mathfrak{gl}(\infty)$, i.e.

\begin{equation}\label{2co}
\hat{r}([a,b]) = [\hat{r}(a),\hat{r}(b)]+c(a,b), \quad\forall a,b\in\mathfrak{gl}(\infty).
\end{equation}

The 2-cocyle is bilinear and different choice of which gives different \textit{central charge}. In our case, we choose

$$
c(E_{i,j},E_{j,i}) = -c(E_{j,i},E_{i,j}) = 1,\quad i<0,j>0
$$
and vanishes elsewhere. Denote $\Lambda_i$ the matrix with all entries in i-th diagonal all equals to 1. Then 

$$
c(\Lambda_i,\Lambda_j) = i\cdot\delta_{i,-j}.
$$

Recall, any projective representation can be regard as a linear representation of a central extension of the original Lie algebra. In such a view of point, we will have $\hat{r}(\lambda_i)=\alpha_i$, yet we will not take this view of point.\par

More precisely, denote the diagonal part of $a\in\mathfrak{gl}(\infty)$ as $\diag(a)$ and $\ndiag(a) = a-\diag(a)$,  then the projective representation on $\Lambda_{0}^{\frac{\infty}{2}}V$ is defined via \cite{kac2013bombay}

$$
\hat{r}(a)\cdot  (\underline{x^{i_1}}\wedge \underline{x^{i_2}}\wedge \underline{x^{i_3}}\wedge\cdots) = r(\ndiag(a)) \cdot (\underline{x^{i_1}}\wedge \underline{x^{i_2}}\wedge \underline{x^{i_3}}\wedge\cdots) + \sum_{i=1}^{\infty}(\lambda_{i_1+\half}-\lambda_{-i+\half})(\underline{x^{i_1}}\wedge \underline{x^{i_2}}\wedge \underline{x^{i_3}}\wedge\cdots) ,
$$
where $\diag(a) =\diag(\cdots,\lambda_\half,\lambda_{-\half},\cdots)$. Especially when acting on $v_{\emptyset}$, we have

$$
\hat{r}(a) = \hat{r}(\ndiag(a)).
$$

Now we can define the projective representation $\hat{R}$ by exponential $\hat{r}$, i.e. $\hat{R}(e^{a}) = e^{\hat{r}(a)}$, and the 2-cocyle $c$ in Lie algebra cohomology determines a 2-cocyle $C$ in the Lie group cohomology, i.e.

$$
\hat{R}(AB) = C(A,B) \hat{R}(A)\cdot\hat{R}(B).
$$

The explicit formula of $C(A,B)$ can be calculated via Baker-Campbell-Hausdorff formula, for $A = e^a,\,B=e^b$. For generic case, this is rather complicated, yet we have the following lemma

\begin{lemma}\label{puu}
For a pair of upper, or lower, triangular matrix $A,B\in GL(\infty)$, if one of them is unitriangular, we have $C(A,B)=1$. 
\end{lemma}
\begin{proof}
Without losing generality we assume $A$ is upper unitriangular, then $\ln A$ is upper triangular with $\diag(\ln A) = 0$. Then by Eq.~(\ref{2co}), we have

$$
c(\ln A,\ln B) = c(m(\ln A,\ln B),n(\ln A,\ln B)) = 0,
$$
where $m,n$ is any commutators involve $\ln A,\ln B$. Since all the contribution of $\ln C$ will come into this form, we have

$$
C(A,B) = 1.
$$
\end{proof}

Any infinite dimensional matrix can be divided into four parts (four quadrants), we define the four parts for arbitrary matrix $A$ as

$$
A^{++} = A_{i,j<0},\quad A^{-+} = A_{i>0,j<0}, \quad A^{+-} = A_{i<0,j>0},\quad A^{--} = A_{i,j>0}.
$$

\begin{remark}
In what follows we will sometimes omit the representation symbol $\hat{r}$, whenr it is clear whether  we are dealing with the algebra or its projective representation. 
\end{remark}
%%%%%%%%%%%%%%%%%%%%%%%%%%%%%%%%%%%%%%%%%%%%%%%%%%%%%%%%%%%%%%%%%%%%%%%%%%%%%%%%%%%

\subsection{Heisenberg, Virasoro, and $W_{1+\infty}$ algebra}

The Heisenberg algebra  $\mathcal{H}$ and Virasoro algebra $\mathcal{V}$ are subalgebras of $W_{1+\infty}$ algebra, which in turn a subalgebra of $\hat{\mathfrak{gl}}(\infty)$. Generators of  $\mathcal{H}$ are $\alpha_{k},\,k\in\Z$, which satisfy the following commutation relation
\begin{equation}\label{commualpha}
[\alpha_i,\alpha_j] = i\delta_{i,-j}
\end{equation}
in the linear representation $\hat{r}$, we have
$$
\alpha_{k} = \sum_{i\in\Z+\half}:\psi_{i-k}\psi^*_i: = \sum_{i\in\Z+\half}i^0:\psi_{i-k}\psi^*_i:.
$$

	Via  the coherent state 
$$
\langle \mathbf p | := \langle 0 | e^{ \sum_{k>0} \frac{1}{k} p_k  \alpha_k},
$$   we  define 
the map $\iota$  as follows.
For  $|V\rangle \in  \Lambda_0^{\frac{\infty}{2}}V$, let
$$
\iota ( |V\rangle):= \langle \mathbf p |V\rangle \    \in   \   \mathbb C[\mathbf p] := \mathbb C[p_1,\cdots,p_k,\cdots]
$$
The map $\iota$ gives a  linear isomorphism
\begin{equation}  \label{bfcorrespondence}
\iota: \ \Lambda_0^{\frac{\infty}{2}}V \longrightarrow  \mathbb C[\mathbf p].
\end{equation}
Hence one can identify the space $\Lambda_0^{\frac{\infty}{2}}V$ with $\mathbb C[\mathbf p]$, which can be considered as the ring of symmetric polynomials, with $p_k$ being the Newton polynomials.

On the other hand the generator $L_{k},\,k\in\Z$ of  Virasoro algebra $\mathcal{V}$ satisfy the commutation relation

\begin{equation}\label{commul}
[L_k,L_m] = (k-m)L_{k+m}+\frac{1}{12}\delta_{k,-m}(k^3-k),
\end{equation}
which represented in $\hat{r}$ as
$$
L_{k} = \sum_{i\in\Z+\half}i:\psi_{i-k}\psi^*_i:.
$$

One can easily notice that $H=L_0$ and $C = \alpha_0$. More generally we have the following operators

\begin{equation*}
W_{r}^{s} = \sum_{i\in\mathbb{Z}+\frac{1}{2}}i^s:\psi_{i-r}\psi_{i}^*:.
\end{equation*}
These operators form a basis of the $W^{1+\infty}$ algebra, and their general commutation relations will be calculated in Appendix II. We will call the $W^r_s$ is in level r, and  $W^{1+\infty}$ algebra is graded under the level, where $\mathcal{H}$ and  $\mathcal{V}$ are the level 0 and level 1 component, resp. Moreover, the level 2 operators are the  cut-and join type of operators used in the Hurwitz theory\cite{goulden1997transitive}. More details about this subsection can be found in \cite{kac1994infinite,alexandrov2015enumerative}, one need to be careful that our convention is slightly different from those in \cite{alexandrov2015enumerative}, our $W_s^r$ corresponding to the $\widetilde{W}^{(r+1)}_s$ in \cite{alexandrov2015enumerative} and $W^{(r+1)}_s$ in\cite{ADKMV}, such a change of convention will appear to be more handy.\par

\begin{remark}
Normally, the commutation relation depends on the choice of $\hat{r}(c)$ used in the linear representation of $\widehat{\mathfrak{gl}}(\infty)$,  which is called the central charge. However, in this work we will always choose $\hat{r}(c)=1$, which result in the commutation relation Eq.~(\ref{commualpha},\ref{commul}).
\end{remark}
%%%%%%%%%%%%%%%%%%%%%%%%%%%%%%%%%%%%%%%%%%%%%%%%%%%%%%%%%%%%%%%%%%%%%%%%%%%%%%%%%%%%%%%%%%%%%%%%%%%%%%%%%%
\subsection{Baker-Akhiezer function} \label{bafunction}

We recall the definition of the Fermionic fields  
\begin{equation*}
\psi(z) := \sum_{k\in\mathbb{Z}+\frac{1}{2}}z^{k-\frac{1}{2}}\psi_k,\quad \psi^*(z):=\sum_{k\in\mathbb{Z}+\frac{1}{2}}z^{-k-\frac{1}{2}}\psi^*_j,
\end{equation*}

The standard Baker-Akhiezer  function is defined by \cite{babelon2003introduction}:
\begin{equation}\label{BAfunc}
\Psi(t,x):= \frac{\left<1|e^{t\alpha_1} \psi(x)   |V\right>}{\left<0| e^{t\alpha_1} |V\right>}.
\end{equation}

Recall the definition of Vertex operator as
\begin{equation*}
\Gamma_{\pm}(\mathbf t):=\exp\left(\sum_{n=1}^{\infty}t_k\alpha_{\pm k}\right),
\end{equation*}
From the definition, one has immediately, the $\Gamma_{\pm}(\mathbf t)$ operator fix the vacuum and covacuum respectively.\par
From Kac\cite{kac1994infinite} (or see also Okounkov 2001\cite{okounkov2001infinite}), we can change the Fermionic field in the definition of Baker-Akhiezer function to vertex operators
%\begin{equation}\label{ftv}
%\begin{split}
%&\psi(z) = z^{C}R\Gamma_{-}(\{z\})\Gamma_{+}(-\{z^{-1}\})\\
%&\psi^*(z) = R^{-1}z^{-C}\Gamma_{-}(-\{z\})\Gamma_{+}(\{z^{-1}\}),
%\end{split}
%\end{equation}
\begin{equation}\label{badefh}
%\begin{split}
\Psi(t,x) = e^{tx}\frac{\left<1|\Gamma(t-\{x^{-1}\})  |V\right>}{\left<0| e^{t\alpha_1} |V\right>}
%\Psi^*(t,x)&= e^{-ty}\frac{\left<-1|\Gamma(-t+\{y^{-1}\}) \psi(x)   |V\right>}{\left<0| e^{-t\alpha_1} |V\right>},
%\end{split}
\end{equation}
where $t-\{x^{-1}\}:=(t-x^{-1},-\frac{x^{-2}}{2},-\frac{x^{-3}}{3},\cdots)$. 
Hence the above definition of the Baker-Akhiezer function \eqref{BAfunc} matches with the definition \eqref{kpprin} in the introduction.

%As we have seen in the statement of Theorem 2.1, there can be an additional unstable contribution in front of the Baker-Akhiezer function, we will define the wave-function with this term included. i.e.

%\begin{equation*}
%\begin{split}
%\Phi(t,x) = e^{S_\un(t,x)/\hbar}\Psi(t,x),
%\Phi(t,x) &= e^{S_\un(t,x)/\hbar}\Psi(t,x) ,
%\end{split}
%\end{equation*}
%where the unstable contribution  defined as
%$$
%S_\un(t,x)/\hbar = -x\ln x+x-t\hbar x.
%$$
%and
%$$
%S^*_\un(t,y/\hbar) = y\ln (-y)-y+t\cdot y,
%$$

%%%%%%%%%%%%%%%%%%%%%%%%%%%%%%%%%%%%%%%%%%%%%%%%%%%%%%%%%%%%%%%%%%%%%%%%%%%%%%%%%%%%%%%%%%%%%%%%%%%%%%%%%%
\subsection{$\E_k$ and $\mathbf{A}[k]$ operator}
Two kinds of operators $\E_k$ and $\mathbf{A}[k]$, first defined in \cite{okounkov2006gromov}, serve a great importance in the operator formalism of Hurwitz theory and Gromov-Witten theory.\par
The $\E_k(z)$ operator is defined as

\begin{equation}\label{defE}
\mathcal{E}_r( z) = \sum_{k \in\mathbb{Z}+\frac{1}{2}}e^{ z(k-\frac{r}{2})}E_{k-r,r}+\frac{c\delta_{0,r}}{\zeta(z)},
\end{equation}

where $\zeta(z) = e^{z/2}-e^{-z/2}$.\par

The definition of the $\mathbf{A}$ operator is more involved, in this paper we will follow those defined in \cite{johnson2009equivariant}, which is defined for equivariant theory with equivariant parameter $\lambda$ and general target $W\PP^{1}(r,s)$, in the case $\gcd(r,s)=1$, which corresponding to the case with trivial gerbe structure. We have

\begin{equation}\label{defAnt}
\mathbf{A}_{a/r}(z,\hbar,\lambda) =
\left\{\begin{aligned}
&\frac{1}{\hbar}\sum_{k=1}^{\infty}\Sc(r\hbar)^{\frac{\lambda z}{r}}\frac{[\lambda z\Sc(r\hbar z)]^k}{(1+\frac{\lambda z}{r})_k}\E_{rk}(\hbar z),&\quad a=0\\
&\frac{1}{\hbar}\frac{\lambda rz}{a}\Sc(r\hbar z)^{\frac{tz+a}{r}}\sum_{k\in\mathbb{Z}}\frac{[\lambda z\Sc(r\hbar z)]^k}{(1+\frac{\lambda z+a}{r})_k}\E_{rk+a}(\hbar z),&\quad a \neq 0.
\end{aligned}\right.
\end{equation}

which have the following non-equivariant limit for $a/r=0/1$,

\begin{equation*}
\mathbf{A}_{0/1}(z,\hbar) =
%\left\{\begin{aligned}
\frac{1}{\hbar}\sum_{k=1}^{\infty}\frac{[\zeta(\hbar z)]^k}{k!}\E_{k}(\hbar z).
%&\frac{1}{\hbar}\frac{rz}{a}\Sc(r\hbar z)^{\frac{a}{r}}\sum_{k\in\mathbb{Z}}\frac{[z\Sc(r\hbar z)]^k}{(1+\frac{a}{r})_k}\E_{rk+a}(\hbar z),&\quad a \neq 0.
%\end{aligned}\right.
\end{equation*}
We will prove it in the next subsection from Eq.~(\ref{defAnt}), with a slight redefinition. We will define $\mathbf{A}_{a/r}[i] = [z^{i+1}]\mathbf{A}_{a/r}(z,\hbar)$, and simply$\mathbf{A}[i]$ for the case $a=0$.\par

All of the operators should be viewed as elements in $\widehat{\mathfrak{gl}}(\infty)$. One might wonders about the appearance of the  term $\frac{c}{\zeta(\hbar z)}$ in the definition of $\E_0(z)$, which serves as a regularization as explained in \cite{okounkov2006gromov}. \par

To understand this regularization term better, we recall a well-known result\cite{kac1994infinite,okounkov2006equivariant,dunin2017quantum}
\begin{lemma}\label{proeigen}
For any operator $\sum_{n\in\mathbb{Z}+\frac{1}{2}}a_n:\psi_n\psi^*_n:$ it has $v_{\lambda}$ as its eigenvector, with eigenvalue 
\begin{equation*}
\sum_{i=1}^{\infty}(a_{\lambda_i-i+1/2}-a_{-i+1/2}),
\end{equation*}
where $\lambda_i$'s with $i>l(\lambda)$ are understood as 0.
\end{lemma}
\begin{proof}
A direct consequence from the definition of Fermionic bilinear.
\end{proof}\par
From above lemma, one can see that $\E_0$ operator can be regarded as the \textit{generating operator} of the $W_r^s$ operators, or when acting on a state serves as generating function of the specific Newton polynomials corresponding to vector in $\Lambda_0^{\frac{\infty}{2}}V$. i.e.

$$
\E_0(z)\left|\lambda\right> = \sum_{k=0}^{\infty}\frac{1}{k!}W_0^k\left|\lambda\right> + \frac{1}{\zeta(z)}\left|\lambda\right> = \sum_{k=0}^{\infty}\sum_{i=1}^{\infty}\frac{(\lambda_i-i+\half)^k}{k!}.
$$ 

where we have used the identity
$$
\sum_{k=0}^{\infty}\sum_{i=1}^{\infty}\frac{(-i+\half)^k}{k!} = \frac{1}{\zeta(z)}.
$$

The $\E_k(z)$ operators satisfy the following commutation relation

\begin{equation}\label{sommu}
[\E_a(z),\E_b(w)] = \zeta(aw-bz)E_{a+b}(z+w).
\end{equation}
%%%%%%%%%%%%%%%%%%%%%%%%%%%%%%%%%%%%%%%%%%%%%%%%%%%%%%%%%%%%%%%%%%%%%%%%%%%%%%%%%%%%%%%%%%%%%%%%%%%%%%%%%%
\subsection{Orbifold Gromov Witten invariants}
The definition of the $\mathbf{A}[k]$ operator enforces us  to consider Gromov-Witten invariants with the position of marked points being remembered. Meanwhile, one also need to consider the connected Gromov-Witten invariants, by connected, we means the domain curves are restricted to be connected.\par

The connected n-point orbifold Gromov-Witten correlator of $\PP[r]$ is defined as
%\begin{equation*}
%G_{g,d,\mathfrak{r},\mathrm{s}}^{\circ}(z_1,\cdots,z_n,w_1,\cdots,w_m):=\prod_{i=1}^{n}z_i^{-1}\prod_{j=1}^{n}w_j^{-1}\int_{[\overline{\mathcal{M}}_{g,d,\mathfrak{r},\mathrm{s}}(W\mathbb{P}(s,r),d)]^{\mathrm{Vir}}}\prod_{i=1}^{n}\frac{\mathrm{ev}^{*}(\mathbf{0}_{r_i})}{1-z_i\bar{\psi}_i}\prod_{j=1}^{m}\frac{\mathrm{ev}^{*}(\mathbf{\infty}_{s_j})}{1-w_j\bar{\psi}_j}.
%\end{equation*}
%%
%%
%In this paper will mainly focus on the $m=0$ case, where the above definition reduce to 
%\begin{equation*}
%G_{g,d,\mathfrak{r}}^{\circ}(z_1,\cdots,z_n):=\prod_{i=1}^{n}z_i^{-1}\int_{[\overline{\mathcal{M}}_{g,d,\mathfrak{r}}(W\mathbb{P}(s,r),d)]^{\mathrm{Vir}}}\prod_{i=1}^{n}\frac{\mathrm{ev}^{*}(\mathbf{0}_{r_i})}{1-z_i\bar{\psi}_i}
%\end{equation*}

\begin{equation}\label{gwgener}
G_{g,d}^{\circ}(z_1,\cdots,z_n):=\prod_{i=1}^{n}z_i^{-1}\int_{[\overline{\mathcal{M}}_{g,n,\emptyset}(\PP[r],d)]^{\mathrm{Vir}}}\prod_{i=1}^{n}\frac{\mathrm{ev}^{*}(\omega)}{1-z_i\bar{\psi}_i}.
\end{equation}

Where we have already took the specialization to the case we are interested in, i.e. the target space is $\PP[r]$, and all other orbifold insertions will not appears, except the encounter of hyperplane class $h$.
Since in the case of $\mathbb{P}[r]$, there is only one point on the target space admits an orbifold structure, the $h$ is the hyperplane class of the $H^*(\mathbb{P}^1)$, or more precisely the pushforward $i_*(\omega)$ via the inclusion
\begin{equation*}
\begin{aligned}
i:\mathbb{P}^1&\rightarrow\mathcal{IP}_{r}\\
p&\rightarrow(p,e),
\end{aligned}
\end{equation*}
i.e. the generator for the component isomorphic to $H^*(\mathbb{P}^1)$. Meanwhile for the orbifold insertions (though not occurs in our case), we define $\phi_{k\neq0}$ as the generators for the r-1 components of $H^*(\mathcal{B}\mathbb{Z}_r)$, where the $\mathcal{B}\mathbb{Z}_r$ is the classifying stack $\mathcal{B}\mathbb{Z}_r:=[\mathrm{pt}/\mathbb{Z}_r]$ for $\mathbb{Z}_r$. Further we have $\tau_i(\phi_j) = \mathrm{ev}_m^*(\phi_j)\bar{\psi}_m^i$.
%, where $\bar{\psi}_m:=c_1(L_m)$, where $L_m$ is the line bundle over the compactified moduli stack $\overline{\mathcal{M}}_{g,n,\mathbf{r}}(\mathcal{P}_{r},d)$, where the fibre is the tangent space at mth marked point considered as locating on the coarse moduli $\PP[r]$ of $\mathcal{P}_{r}$.
% To understand this paper, it is sufficient to keep in mind that the target in orbifold Gromov-Witten theory is not the coarse moduli of the orbifold, but its inertia stack.\par %The generating function $\hat Z (\hbar,\bT)$i s defined in \eqref{gwp} for stationary orbifold Gromov-Witten invariants.\par

%$$
%\hat Z (\mathbf{x},\hbar) = Z (\hbar, \bT)\big|_{T_i \rightarrow x_{i+1}}.
%$$
%where $Z (\hbar, \bT)$ 
%{\Guo: Add the definition of the generating function $\hat Z$ (LHS of Theorem 2.3) here}

%%
%%
We will not deep into the geometric aspect of orbifold Gromov-Witten theory of weighted projective line, more details about Chen-Ruan class and orbifold Gromov-Witten theory can be found in \cite{chen2004new,chen2001orbifold,abramovich2008gromov,johnson2009equivariant}
\begin{remark}\label{r1.4}
Although, we are dealing with an orbifold Gromov-Witten theory, In this paper we will only consider the target as the main component, i.e. the one isomorphic to the coarse moduli $\mathcal{P}_{r}$, and reduce to $\PP^1$ for the case $r=1$.\par
% Since  the partition function of the orbifold Gromov-Witten theory of $\PP[r]$ is a tau-function of Toda-Lattice, when we confine ourselves to this case, roughly speaking, half of the information of the integrable hierarchy will be lost.
\end{remark}

\section{Wave Function and Quantum Curve}\label{s2}

%\begin{remark}
%We will change the Gromov-Witten times $T_i$ in Eq.~(\ref{gwp}) to $x_{i+1}$ in the sequel, which differs from the convention used in \cite{okounkov2006equivariant,johnson2009equivariant}, i.e. our indices are large than theirs by $1$. In the sequel by $t_i$ we will refer  to the KP times.
%\end{remark}
From A-model side of view, the wave function for the orbifold Gromov-Witten (Gromov-Witten) theory of $\mathbb{P}[r]$, or more generally for any enumerative problem of KP/KdV type, is defined by the specialization of the generating function for the Gromov-Witten invariants
\begin{equation*}
\Psi(x) := \hat{Z}(\{z(x)\}) := e^{\hat{F}(\{z(x)\})} :=e^{\sum_{k=0}^{\infty}\hat{S}_k(z(x))\hbar^{k-1}},
\end{equation*}
where $\{z(x)\} = \{z(x),z(x),\cdots\}$, and 
\begin{equation*}
S_k(z(x)) = \sum_{2g-2+n=k-1}\frac{1}{n!}G^\circ_{g,n}(z(x),z(x),\cdots,z(x)).
\end{equation*}
The disconnected 1-point function is defined in Eq.~(\ref{gwgener}), and $z(x)$ is a choice of local coordinates, transform the Gromov-Witten partition function to the Baker-Akhiezer function of KP hierarchy. \par
We note in the definition of $S_0$ and $S_1$, it allows one to include the unstable contributions, that is why we used $\hat{Z}$ instead of $Z$. The choice of which is quite crucial, since as conjectured in \cite{dumitrescu2013spectral}, a proper choice of the unstable terms in $S_0$ and $S_1$ will fully determine the classical spectral curve of the underlying enumerative geometric problem, i.e. its mirror LG model.
\par
On the other hand, from the B-model side of view, one can define the generating function $Z$ solely from the LG potential, which is the definition equation of the so-called spectral curve, or mirror curve. Such procedure has been known as the Eynard-Orantin topological recursion \cite{eynard2007invariants}, moreover it is conjectured\cite{gukov2012polynomial} the wave function $\Psi(x)$, specializing in a proper way, contains in the kernel of some differential (or, difference) operator
\begin{equation*}
\hat{H}(x,-\hbar\partial_x)\Psi(x)=0,
\end{equation*}
where $\hat{H}(x,-\hbar\partial_x)$ is the so-called quantum curve, which generates a holonomic system determines the B-model partition function\cite{ADKMV,dijkgraaf2008supersymmetric,dijkgraaf2009quantum,mulase2012spectral} and can be viewed as a quantization of the classical spectral curve, and will return to its classical form when taking the classical limit, i.e.
\begin{equation*}
\hat{H}(x,-\hbar\partial_x) = H(x,y(x)), \quad \hbar\rightarrow 0.
\end{equation*}
However, it is not always the case that $\hat{H}_t(x,-\hbar\partial_x) = H(x,-\hbar\partial_x)$, since there may exist higher order corrections of $\hbar$.
%The obstruction of such holonomic structure is an element in $K_2$ of the algebraic K-theory\cite{gukov2012polynomial}.\par
%%
%%
When the underlining Gromov-Witten theory (A-model) has a mirror LG theory (B-model), the generating functions we talked above, coincide with each other, as in the case of $W\mathbb{P}^1(s,r)$\cite{tang2017equivariant,fang2017eynard}. Then the wave functions also match, which suggests that one can consider the quantum curve directly from A-model.\par
Another natural definition of A-model wave function arise from the Baker-Akhiezer function. It is known, some specific Gromov-Witten theory (or other enumerative geometric problems), are naturally linked to the theory of integrable system. More precisely, the generating function of the geometric theories is a tau-function of some integrable hierarchy. Which suggests a natural link between the wave function and the Baker-Akhiezer function of the integrable hierarchy, since both of them corresponding to the asymptotic expansion around a boundary of Riemann surface. This definition of wave function makes the calculation much simpler.\par

Therefore, our approach in this section will be, calculate the one-point function of A-model by  Baker-Akhiezer function of the corresponding KP hierarchy, and check whether it is located inside the kernel of the quantum curve obtained from B-model.\par

% which will further become an identification via a proper choice of the local coordinates $z(x)$.\par
%%
%%
%Whats more, the quantum curve for the Gromov-Witten theory can be related to an operator, which annihilate the Baker-Akhiezer function of the integrable hierarchy (we will call it the quantum curve for the integrable hierarchy). This view of point will play an essential role, when we calculate the two-point function in the next section.\par
%%
%%
Now we will first derive the explicit expression for the wave function of the orbifold Gromov-Witten theory of $\mathbb{P}[r]$, more precisely, we will prove the following theorem
\begin{theorem}\label{onept}
The wave function for the orbifold Gromov-Witten theory of $\mathbb{P}[r]$ has the following closed form
\begin{equation*}
\Phi(t,x) = e^{S_\un(t,x)/\hbar}\Psi(t,x) \sim\left(2\pi\right)^{\half} \sum_{d=0}^{\infty}\frac{(-1)^dq^{rd}e^{t\hbar rd}}{r^dd!}\frac{\hbar^{-\frac{x}{\hbar}-(r+1)d}}{\Gamma(\frac{x}{\hbar}+rd+\frac{1}{2})},
\end{equation*}
where the unstable contribution is  defined as
$$
S_\un(t,x) = -x\ln x+x-t\hbar x.
$$
%The counterpart for the adjoint wave function is
%\begin{equation*}
%\Phi^*(t,y) = e^{S^*_\un(t,y)/\hbar}\Psi^*(t,y) \sim \left(2\pi\right)^{\half}\sum_{d=0}^{\infty}\frac{(-1)^dq^{rd}e^{t\hbar rd}}{r^dd!}\frac{\hbar^{\frac{y}{\hbar}-rd+\frac{1}{2}}}{\Gamma(-\frac{y}{\hbar}+rd-\frac{1}{2})},
%\end{equation*}
%with the unstable contribution
%$$
%S^*_\un(t,y) = y\ln (-y)-y+t\cdot y,
%$$
where $\sim$ refers to the asymptotic expansion at $x\rightarrow\infty,|\arg(x)|<\pi$. For the case $r=1$ there will be an additional normalization factor $e^{-\frac{qe^{t\hbar}}{\hbar^2}}$.
\end{theorem}\par

The following definition will be used in the future
\begin{equation*}
\Phi^d(x) := \frac{\left(2\pi\right)^{\half}(-1)^d}{r^dd!}\frac{\hbar^{-\frac{x}{\hbar}-(r+1)d}}{\Gamma(\frac{x}{\hbar}+rd+\frac{1}{2})}.
\end{equation*}
which means Theorem 2.1 is equivalent to say
\begin{equation}
\Phi(t,x)\sim e^{S_\un(t,x/\hbar)}\sum_{d=0}^{\infty}q^{rd}e^{t\hbar rd}\Phi^d(x).
\end{equation}\par
There is also another important way to rewrite the wave functions that we need to introduce, since it is essential for the proof of Theorem 2.1. 
\begin{corollary}\label{rr}
The wave function can also be put into the following form
\begin{equation*}
%\begin{split}
\Phi(t,x)= e^{S_\un(t,x)/\hbar}\left(\left(2\pi\right)^{\half}e^{S_{\mathrm{ex}}(x)}\sum_{d=0}^{\infty}\frac{(-1)^dq^{rd}e^{t\hbar rd}}{r^dd!\hbar^{(r+1)d}}\prod_{i=1}^{rd}\frac{\hbar}{x+(i-\frac{1}{2})\hbar}\right):=e^{t\hbar x+\si(x)}\sum_{d=0}^{\infty}\tilde{\Phi}_d,
%\Phi^*(t,y)&= e^{S^*_\un(t,y/\hbar)}\left(\left(2\pi\right)^{\half}e^{S_{\mathrm{ex}}^*(y)}\sum_{d=0}^{\infty}\frac{(-1)^dq^{rd}e^{t\hbar rd}}{r^dd!\hbar^{(r+1)d}}\prod_{i=0}^{rd-1}\frac{\hbar}{-y+(i-\frac{1}{2})\hbar}\right):=e^{ty+\si^*(y)}\sum_{d=0}^{\infty}Y_d,
%\end{split}
\end{equation*}
Where $\si(x)$ is a singular function, which is related to the asymptotic expansion of Gamma functions
\begin{equation*}
%\begin{split}
\si(x) \sim -\ln\Gamma(\frac{x}{\hbar}+\frac{1}{2})- \frac{x}{\hbar}\ln\hbar +\frac{1}{2}\ln(2\pi),
%&\si^*(y)\sim-\ln\Gamma(-\frac{y}{\hbar}-\frac{1}{2}) +\left(\frac{y}{\hbar}+\frac{1}{2} \right)\ln \hbar+\frac{1}{2}\ln(2\pi),
%\end{split}
\end{equation*}
for $x\rightarrow\infty,|\arg(x)|<\pi$. Further we have defined
\begin{equation*}
\tilde{\Phi}_d := \frac{(-1)^dq^{rd}e^{t\hbar rd}}{r^dd!\hbar^{(r+1)d}}\prod_{i=1}^{rd}\frac{\hbar}{x+(i-\frac{1}{2})\hbar},\quad d\in\mathbb{N}.
\end{equation*}

\end{corollary}\par

The explicit expression for $S_{\mathrm{ex}}(x)$  can be read off from the definition of $\si(x)$, yet we will not need them in the sequel.
\begin{remark}
Since the variable $x$ is purely formal, from now on, by a slight misuse of notation, we will sometimes write $=$, although we actually refer to the asymptotic expansion at $x\rightarrow\infty,|\arg(x)|<\pi$ .
%Also, from now on we will assume working at some fixed $r>1$, and save the $r$ subscripts whenever it is clear.
\end{remark}
%%
%%
%Meanwhile $\Phi(x),\Phi^*(y)$ are the Baker-Akhiezer function and adjoint Baker-Akhiezer function, resp., for the KP hierarchy induced by the generating function of Gromov-Witten theory, when viewed as a tau-function of Toda-Lattice hierarchy.\par
%%
%%
Our second aim will be proving 
\begin{theorem}\label{t2.2}
The wave function of the Gromov-Witten theory of $\mathbb{P}[r]$ satisfy the following quantum curve equation
$$
\hat{H}_t(x,-\hbar\partial_x)\cdot\Phi(t,x) = 0.
%\quad \hat{H}^*(y,\hbar\partial_y)\cdot\Phi^*(t,y) = 0.
$$
where
\begin{equation*}
%\begin{split}
\hat{H}_t(x,-\hbar\partial_x) = e^{-\hbar\frac{\partial}{\partial x}}+q^re^{t\hbar rd}e^{r\hbar\frac{\partial}{\partial x}}-x+\frac{\hbar}{2}.
%\hat{H}_t^*(y,\hbar\partial_y) = e^{\hbar\frac{\partial}{\partial y}}+q^re^{t\hbar rd}e^{-r\hbar\frac{\partial}{\partial x}}-y-\frac{\hbar}{2}.
%\end{split}
\end{equation*}
\end{theorem}\par

From the definition of quantum curve  we can see it is 'quantized' from the classical spectral curve, i.e. $\hat{H}_t(x,-\hbar\partial_x) = H(x,-\hbar\partial_x)+O(\hbar)$, where $O(\hbar)$ is called \textit{quantum correction}.
%%
%%
%The counterparts for the corresponding KP hierarchy will be
%\begin{equation*}
%\hat{H}_{KP}(x,\hbar\partial_x) = D_{x}^{-1}+\frac{q^{r}}{\hbar^{r+1}}D_x^{r} - x
%\end{equation*}
%and
%\begin{equation*}
%\hat{H}^*_{KP}(x,\hbar\partial_x) = D_{y}^{-1}+\frac{q^{r}}{\hbar^{r+1}}D_y^{r} - y
%\end{equation*}
%where
%\begin{equation*}
%D_x = \exp(\hbar\partial_x)\left(\frac{x}{\hbar}\right)^{-1}
%\end{equation*}
%and its adjoint w.r.t. the Hermitian inner product $\mathrm{Res}_{t=0}$ on $\mathbb{C}[[t,t^{-1}]]$
%\begin{equation*}
%D^*_y = \left(\frac{y}{\hbar}\right)^{-1}\exp(-\hbar\partial_y)
%\end{equation*}
%%
\begin{remark}\label{rrr}
For the rest of the section we will work in the case $t=0$, in order to avoid unnecessary complication. In another word, we will calculate the $\Psi(0,x)$, which will be simply denoted by $\Psi(x)$. The prove will be complete after Section \ref{s3.1}, where we will see, the generalization to $t\neq0$ is trivial, by simply doing the rescaling $q\mapsto qe^t$, and a modify the unstable term by $t\hbar x$. We will use prime to 
\end{remark}

%%%%%%%%%%%%%%%%%%%%%%%%%%%%%%%%%%%%%%%%%%%%%%%%%%%%%%%%%%%%%%%%%%%%%%%%%%%%%%%%%%%%%%%%%%
%%%%%%%%%%%%%%%%%%%%%%%%%%%%%%%%%%%%%%%%%%%%%%%%%%%%%%%%%%%%%%%%%%%%%%%%%%%%%%%%%%%%%%%%%%
\subsection{Operator formalism for the Gromov-Witten theory of $\mathbb{P}[r]$}
Our starting point will be the partition function of the equivariant orbifold Gromov-Witten theory of $W\mathbb{P}^1(r,s)$, which similar to those of Okounkov and Pandharipande \cite{okounkov2006gromov,okounkov2006equivariant}, has an operator formalism derived by Johnson \cite{johnson2009equivariant}:
\begin{theorem}[Chap. VI in \cite{johnson2009equivariant}]
The generating function of the equivariant orbifold Gromov-Witten theory of $W\mathbb{P}^1(r,s)$, with equivariant parameter denote by $\lambda$, can be written in the following form
\begin{equation}\label{tautl}
Z(\mathbf x,\mathbf x^*,\hbar,\lambda) =\left<e^{\sum x_i\mathbf{A}[i]}e^{\frac{\lambda\alpha_1}{\hbar}}\left(\frac{q}{\lambda(-\lambda)^{1/r}}\right)^{H}e^{\frac{-\lambda\alpha_{-r}}{r\hbar}}e^{\sum x_j^*\mathbf{A}^*_{r_m}[j]}\right>.
\end{equation}
Furthermore, it is a tau-function of Toda-Lattice (TL) integrable hierarchy.
\end{theorem}\par
The generating function on the LHS of \eqref{tautl}  is defined similar to the stationary one \eqref{gwp}, however, with insertions on the full state space.
In this paper, we are focusing on the stationary case defined in \eqref{gwp}. We will only use the    special  case of the  Eq.~(\ref{tautl}):
\begin{equation}\label{equitau}
Z(\hbar,\bT)= Z(\mathbf x,\mathbf x^*,\hbar,\lambda)\Big|_{x_i=T_{i},x^*_i=0,\lambda=0}= \left<e^{\sum T_i\mathbf{A}[i]}e^{\frac{\lambda\alpha_1}{\hbar}}\left(\frac{q}{t(-t)^{1/r}}\right)^{H}e^{\frac{-\lambda\alpha_{-r}}{r\hbar}}\right>\big|_{\lambda=0}.
\end{equation}
where $\mathbf{A}[i] = [z^{i+1}]\mathbf{A}_{0/1}(z,\hbar,\lambda)$ .\par

%%
%%
%At this point we are still working on the equivariant orbifold Gromov-Witten theory, 
To proceed we need to derive the non-equivariant limit of the rhs of Eq.~(\ref{equitau}). We have
\begin{corollary}\label{c2.1}
The partition function of the orbifold Gromov-Witten theory of $\mathbb{P}[r]$ when restricted to the positive times, is
\begin{equation}\label{taufunc}
Z(\hbar,\bT) = \left<e^{\sum T_i\mathbf A[i]}e^{\frac{\alpha_1}{\hbar}}q^{H}e^{\frac{\alpha_{-r}}{r\hbar}}\right>,
\end{equation}
where $\mathbf A[i] = [z^{i}]\mathbf A(z,\hbar)$, and 
\begin{equation}\label{Atoalpha}
\mathbf A(z,\hbar) = \frac{1}{\hbar}\sum_{k=0}^{\infty}\frac{(z\mathcal{S}(\hbar z))^k}{k!}\mathcal{E}_{k}(\hbar z).
\end{equation}
By the following change of variables
\begin{equation}\label{ttox}
\tau(\hbar,\bp):=Z(\hbar,\bT) |_{T_{k} \mapsto k! p_{k+1}}.
\end{equation}
the partition function becomes a tau-function of KP hierarchy, with KP times $t_k:= p_k/k$.
%We will call this specialization stationary orbifold  Gromov-Witten theory of $\PP[r]$.
\end{corollary}

\begin{proof}
First one can see that the $(-\lambda)^{1/r}$ in the denominator of $\left(\frac{q}{\lambda(-\lambda)^{1/r}}\right)^{H}$ can be cancelled with the $-\lambda$ in $e^{\frac{-\lambda\alpha_{-r}}{r\hbar}}$, therefore we have
\begin{equation}\label{eq39}
Z(\hbar,\bT) = \left<e^{\sum T_i\mathbf{A}[i]}e^{\frac{\lambda\alpha_1}{\hbar}}\left(\frac{q}{\lambda}\right)^{H}e^{\frac{\alpha_{-r}}{r\hbar}}\right>.
\end{equation}\par

Recall the definition for $\mathbf{A}_{0/1}(z,\hbar,\lambda)$ from Eq.~(\ref{defAnt})
\begin{equation*}
\mathbf{A}_{0/1}(z,\hbar,\lambda)=\frac{1}{\hbar}\mathcal{S}(\hbar z)^{\lambda z}\sum_{k\in\mathbb{Z}}\frac{[\lambda z\mathcal{S}(\hbar z)]^k}{(1+\lambda z)_k}\mathcal{E}_{k}(\hbar z),
\end{equation*}
where operator $\mathcal{E}_{k}(\hbar z)$ have energy $-k$, i.e.
\begin{equation*}
[H,\mathcal{E}_{k}(\hbar z)] = -k\mathcal{E}_{k}(\hbar z),
\end{equation*}
thus its non-equivariant limit will give us a prefactor $\frac{\lambda ^k}{\hbar k!}\propto \lambda ^k$, and we have $\alpha_1 = \mathcal{E}_{1}(0)$ has energy -1. Therefore we have for energy $-k$ contributions of the operator
\begin{equation*}
e^{\sum T_i\mathbf{A}[i]}e^{\frac{\lambda \alpha_1}{\hbar}},
\end{equation*} 
a prefactor in the non-equivariant limit proportional to $t^{k}$. Since the expectation value for 
\begin{equation*}
\left<\mathcal{E}_{a_1}(z_1),\cdots,\mathcal{E}_{a_n}(z_n)\right>,
\end{equation*}
can be non-vanishing only if $\sum_{i=1}^{n}a_i=0$, therefore, for energy $-k$ contribution of $e^{\sum T_i\mathbf{A}[i]}e^{\frac{\lambda \alpha_1}{\hbar}}$, $e^{\frac{\alpha_{-r}}{r\hbar}}$ must contribute a energy k contribution, thus for $\left(\frac{q}{\lambda }\right)^{H}$ we have a $\lambda ^{-k}$ contribution, which means we have a well-defined non-equivariant limit for Eq.~(\ref{equitau}). 
%Changing the notation $\mathbf{A}$ to $A$ gives the desired formula.

%%
%%

%%
%%
In order to derive the coordinate changing formula, we will briefly outline the prove about how this generating function can be view as a KP tau-function.\par
Following \cite{okounkov2006equivariant,johnson2009equivariant}, one can see there exist an upper unitriangular dressing matrix $W$, for which we have
\begin{equation*}
W^{-1}e^{\sum T_i\mathbf A[i]}W = e^{\sum T_i\frac{\hbar}{i!}\alpha_{i}},
\end{equation*}
viewing as element in $GL(\infty)$. The existence of such operator can be easily seen by
\begin{equation*}
\mathbf A[i] = \frac{1}{ \hbar (i+1)!}\alpha_{i+1}+\cdots,
\end{equation*}
where the dots stands for the higher energy contributions ($\alpha_k,\, k<i$), and the factor $\frac{1}{\hbar^ii!}$ is derived from Eq.~(\ref{Atoalpha}) by letting $\hbar\rightarrow0$. For the justification of such operation, which will need the monomiality of the coefficients, one is referred to \cite{okounkov2006gromov,okounkov2006equivariant}
Therefore we have
\begin{equation*}
Z(\hbar,\bT)=  \left<We^{\sum T_i\frac{\hbar}{i!}\alpha_i}W^{-1}e^{\frac{\alpha_1}{\hbar}}q^{H}e^{\frac{\alpha_{-r}}{r\hbar}}\right>,
\end{equation*}
since $W$ is upper unitriangular, it fixes the vacuum, and due to $G=W^{-1}e^{\frac{\alpha_1}{\hbar}}q^{H}e^{\frac{\alpha_{-r}}{r\hbar}}\in \widehat{GL}(\infty)$, the above function is already in the standard form of KP hierarchy $\braket{0|\Gamma_{+}(t)G|0}$. Where
\begin{equation*}
t_i = \frac{T_{i-1}}{\hbar i!},\quad i\geq1.
\end{equation*}
\end{proof}
%%
%%
%%%%%%%%%%%%%%%%%%%%%%%%%%%%%%%%%%%%%%%%%%%%%%%%%%%%%%%%%%%%%%%%%%%%%%%%%%%%%%%%%%%%%%%%%%%%%%%%%%%
\begin{remark}
%There is one issue need to be clarified, The dressing operator $W$ is apparently not unique, different choice of the dressing operator may gives different linear relations between the Gromov-Witten times and KP times, which means the KP hierarchy may flow to another one, however, such two hierarchy should be considered equivalent, the discussion of which can be found in \cite{kharchev1998kadomtsev}
%\color{red}(this sentence is based on my rough understanding)}.
The dressing operator $W$ is only unique up to a left multiplication of a centralizer of $\alpha_1$. The case we have in hand is much simpler compared to those in \cite{okounkov2006equivariant,johnson2009equivariant}, since no positive energy operators ever shows up. A new treatment of dressing operator can be found in \cite{oblomkov2018gw}.
\end{remark}\par
The above $\mathbf A[i]$ operator specialized from \cite{johnson2009equivariant} is different from those specialized from $\mathbf A^{\mathrm{OP}}[i]$ in \cite{okounkov2006equivariant}, and they are related by the following Corollary:\par

\begin{corollary}
The partition function of the orbifold Gromov-Witten theory of $\mathbb{P}[r]$ when restricted to the positive times, can be also written as
\begin{equation}\label{taufunc1}
Z(\hbar,\bT) = \left<e^{\sum T_i\mathbf A^{\mathrm{OP}}[i]}e^{\alpha_{1}}\left(\frac{q}{\hbar}\right)^He^{\frac{\alpha_{-r}}{r\hbar}}\right>,
\end{equation}
where $\mathbf A^{\mathrm{OP}}[i] = [z^{i}]\mathbf A(z,\hbar)$, and 
$$
\mathbf A^{\mathrm{OP}}(z,\hbar) = \frac{1}{\hbar}\sum_{k=0}^{\infty}\frac{\zeta(\hbar z)^k}{k!}\mathcal{E}_{k}(\hbar z).
$$
Which is a tau-function of KP hierarchy, with KP times:
$$
t_i = T_{i+1}\frac{1}{\hbar^{i+1}i!}.
$$
We will call this specialization stationary orbifold  Gromov-Witten theory of $\PP[r]$.
\end{corollary}

\begin{proof}
Since $\left(\frac{1}{\hbar}\right)^{H}$ fixes the vacuum and covacuum, we can insert one of it on the left side inside the braket of Eq.~(\ref{taufunc}) adjunct to the covaccum, without changing its value
\begin{equation*}
Z(\hbar,\bT) =  \left<\left(\frac{1}{\hbar}\right)^{H}e^{\sum T_i \mathbf A[i]}e^{\frac{\alpha_1}{\hbar}}q^{H}e^{\frac{\alpha_{-r}}{r\hbar}}\right>.
\end{equation*}
Now we  consecutively commute this operator to the middle in order for $H$ contributions to merge. Since $\mathcal{E}_r(\hbar z)$ has energy $-r,$ we have
\begin{equation*}
\left(\frac{1}{\hbar}\right)^{H}\mathcal{E}_r(\hbar z)\left(\frac{1}{\hbar}\right)^{-H}=  \hbar^{r}\mathcal{E}_r(\hbar z),
\end{equation*}
which means commute $\left(\frac{1}{\hbar}\right)^{H}$ through an $\mathbf A[i]$ operator amount multiply a $\hbar^{r}$, or equivalently changing $\mathbb A(z,\hbar)$ to the following form
\begin{equation*}
\mathbf{A}^{\mathrm{OP}}(z,\hbar) := \frac{1}{\hbar}\sum_{k=0}^{\infty}\frac{\zeta(\hbar z)^k}{k!}\mathcal{E}_{k}(\hbar z),
\end{equation*} 
which is exactly the operators specialized from those used in \cite{okounkov2006equivariant}, and since $\alpha_1 = \mathcal{E}_1(0)$, one also has
\begin{equation*}
\left(\frac{1}{\hbar}\right)^{H}e^{\frac{\alpha_1}{\hbar}}\left(\frac{1}{\hbar}\right)^{-H} = e^{\alpha_1}.
\end{equation*}
Using these results, we arrive at
\begin{equation*}
Z(\hbar,\bT) =  \left<e^{\sum T_i\mathbf{A}^{\mathrm{OP}}[i]}e^{\alpha_1}\left(\frac{q}{\hbar}\right)^{H}e^{\frac{\alpha_{-r}}{r\hbar}}\right>,\end{equation*}
where $\mathbf{A^{\mathrm{OP}}}[i] = [z^{i+1}]\mathbf{A}^{\mathrm{OP}}(z,\hbar)$.\par
Now, since one has
\begin{equation*}
\mathbf{A}^{\mathrm{OP}}[i] = \frac{\hbar^{i}}{ (i+1)!}\alpha_{i+1} +\cdots,
\end{equation*}
where the dots also stand for higher energy terms. There exist a different upper triangular dressing operator $W'$ s.t.
\begin{equation*}
Z(\hbar,\bT) =  \left<e^{\sum T_i\alpha_i}W^{'-1}e^{\frac{\alpha_1}{\hbar}}\left(\frac{q}{\hbar}\right)^{H}e^{\frac{\alpha_{-r}}{r\hbar}}\right>.
\end{equation*}
Thus by the same argument in the proof of Corollary \ref{c2.1}, we have the partition function above is a tau-function of KP hierarchy, with KP times:
\begin{equation*}
t_i = \frac{T_{i-1}}{\hbar^{i+1}i!},\quad i\geq1.
\end{equation*}
%Remove the superscript $\mathrm{OP}$ will complete the proof.
\end{proof}
%%
%%
%which means the KP hierarchy followed from
%\begin{equation*}
%G=W^{-1}e^{\frac{\alpha_1}{\hbar}}\left(\frac{q}{\hbar}\right)^{H}e^{\frac{\alpha_{-r}}{r\hbar}},
%\end{equation*}
%to 
%\begin{equation*}
%G'=W^{'-1}e^{\alpha_1}\left(\frac{q}{\hbar}\right)^{H}e^{\frac{\alpha_{-r}}{r\hbar}}.
%\end{equation*}
%which will give us a different linear relation between the Gromov-Witten times and KP times.\par
%%
%%
\begin{remark}
In what follows, we will not use this convention from \cite{okounkov2006equivariant}, and we will refer to  Eq.~(\ref{taufunc}) whenever we say tau function. We see from above, there is a degree of freedom for the redefinition of $\mathbf A[i]$ operators. The reason why we choose this specific form is, now the principal specialization is easier to be realized.
\end{remark}
% the reason we choose the specific form in Corollary \ref{c2.1}, is due to we want the expression for the wave function $\Phi(x)$ to be homogenous under the grading
%\begin{equation*}
%\deg(\hbar) = \deg(x) = 1,\quad \deg(q) = 1+\frac{1}{r}.
%\end{equation*}. 

%Which will be justified in the rest of this section.\par
%{\color{red}(This reasoning is not satisfactory, but I cannot come up with a better one.)}
%%
%%
%%%%%%%%%%%%%%%%%%%%%%%%%%%%%%%%%%%%%%%%%%%%%%%%%%%%%%%%%%%%%%%%%%%%%%%%%%%%%%%%%%%%%%%%%%%%%%%%%%%%%%%%%%
\subsection{Connected vev}We note $\sum T_i\mathbf A[i]$ is actually working as a truncation, i.e. remove the negative orders (the unstable contributions) in $ \mathbf A(z,\hbar)$ by brutal force. Therefore to keep the notation clean, we will instead, consider the following pseudo generating series
\begin{equation*}
G(z,\hbar):=\sum_{n=0}^{\infty}\frac{1}{n!}G^\bullet(z_1,\cdots,z_n,\hbar) := \sum_{n=0}^{\infty}\sum_{d=0}^{\infty}\frac{q^{rd}}{n!}G_d^\bullet(z_1,\cdots,z_n,\hbar) ,
\end{equation*}
where we define $\mathcal{A}(z,\hbar) := \hbar \mathbf A(z,\hbar)$, and
\begin{equation}\label{Gdbu}
G_d^\bullet(z_1,\cdots,z_n,\hbar)  = \frac{1}{\hbar^{n}}\left<\prod_{i=1}^{n}\mathcal{A}(z_i,\hbar)e^{\frac{\alpha_1}{\hbar}}P_{rd}e^{\frac{\alpha_{-r}}{r\hbar}}\right>.
\end{equation}
where $P_{rd} = \sum_{|\lambda|=rd}\ket{\lambda}\bra{\lambda}$ is the projection on to the degree $rd$ subspace. The reason that we only consider the degree $rd$ layers, is due to we can rewrite Eq.~(\ref{Gdbu}) as
\begin{equation}\label{Grd}
G_{d}^\bullet(z_1,\cdots,z_n,\hbar)=\frac{1}{d!(rd)!r^{d}\hbar^{(r+1)d+n}}\left<\alpha_1^{rd}\prod_{i=1}^{n}\mathcal{E}_0(\hbar z_i)\alpha_{-r}^{d}\right>,
\end{equation}
by using 
\begin{equation*}
e^{\alpha_1}\mathcal{E}_0(\hbar z)e^{-\alpha_1} = \mathcal{A}(z,\hbar),
\end{equation*}
which is a simple calculation by using the commutation relation Eq.~(\ref{sommu}): $[\mathcal{E}_a(z_1),\mathcal{E}_b(z_2)] = \zeta(az_2-bz_1)\mathcal{E}_{a+b}(z_1+z_2)$. Therefore Eq.~(\ref{Grd}) is non-vanishing only if $d\in\mathbb{Z}$.
Furthermore, we define the connected vacuum expectation value (vev) recursively by
$$
G_d^\bullet(z_1,\cdots,z_n,\hbar)= \sum_{P\in\mathrm{Part}_d[n]}\frac{1}{|\mathrm{Aut}(P)|}\prod_{i=1}^{l(P)}G_{d_i}^{\circ}(z_{P_i},\hbar),
$$
where the automorphism group is always trivial for $r\neq1$. To be complete, we will give a brief account for the following well-known result in Appendix I, which relate the disconnected theories to the connected theories.
\begin{lemma}\label{t2.4}
%After changing the variables in the pseudo generating function $G(z,\hbar)$ into Gromov-Witten times, which we will denote as $G(x,\hbar)$, we have

\begin{equation*}
\begin{split}
\ln G(z,\hbar) = \sum_{n=0}^{\infty}\sum_{d=0}^{\infty}\frac{q^{rd}}{n!d!(rd)!r^{d}\hbar^{(r+1)d+n}}\left<\alpha_1^{rd}\prod_{i=1}^{n}\mathcal{E}_0(\hbar z_i)\alpha_{-r}^{d}\right>^{\circ}.
\end{split}
\end{equation*}
%where $P^*\in\mathrm{Part}^*_d[n]$ includes the following data
%\begin{equation*}
%P^*=\{(d_1,n_1),\cdots,(d_{l(P)},n_{l(P)})\},
%\end{equation*}
\end{lemma}\par
\begin{proof}
See Appendix I.
\end{proof}
Finally, we have the following theorem from \cite{johnson2009equivariant}
\begin{theorem}[Chap. V in \cite{johnson2009equivariant}]
The connected vev is related to the Gromov-Witten invariants by
\begin{equation*}
G_d^\circ(z_1,\cdots,z_n,\hbar) :=\left<\alpha_1^{rd}\prod_{i=1}^{n}{\mathcal{E}_{0}(z_i)}\alpha_{-r}^{d}\right>^\circ= \sum_{g=0}^{\infty}\hbar^{2-2g}G_{g,rd}^\circ(z_1,\cdots,z_n),
\end{equation*}
where $G_{g,rd}^{\circ}(z_1,\cdots,z_n)$ is the orbifold Gromov-Witten generating function defined in Eq.~(\ref{gwgener}).
\end{theorem}\par
As we have mentioned, after a change of coordination defined in appendix, $G(z,\hbar)$ differs from $Z(\mathbf x,\hbar)$ by adding the unstable contributions, which are defined for $\forall r\in\mathbb{Z}_+$ as:
\begin{equation}\label{grading}
G^\circ_{0,0}() = 0,\quad G^\circ_{1,0}() = 0,\quad G^\circ_{0,0}(z_1) = \frac{1}{z_1},\quad G_{0,0}^\circ(z_1,z_2) = 0.
\end{equation}

%%
%%%%%%%%%%%%%%%%%%%%%%%%%%%%%%%%%%%%%%%%%%%%%%%%%%%%%%%%%%%%%%%%%%%%%%%%%%%%%%%%%%%%%
\subsection{Principal Specialization}
Now we want to carry out the  principal specialization in order to get the Baker-Akhiezer function of KP hierarchy by Eq.~(\ref{kpprin}).\par
From Corollary \ref{c2.1}, we already knew how the KP times are related to the Gromov-Witten times, since the principal specialization is defined for KP times via $t_k = -\frac{1}{kx^k}$ combining with Eq.~(\ref{ttox}) in Corollary \ref{c2.1} we have
\begin{equation}\label{pringwtox}
T_{k-1} = -\frac{1}{kx^k}\cdot \hbar k! = -\frac{\hbar(k-1)!}{ x^k}.
\end{equation}
For the descendent times, the principal specialization refers to

\begin{equation}\label{pringwtoy}
z_i^k =  -\frac{\hbar(k-1)!}{ x^k}, \quad\forall i.
\end{equation}

Now we will implement this specialization on $G(z,\hbar)$. From whose very definition, we have
\begin{equation}\label{discon}
G(z,\hbar)=\sum_{n=0}^{\infty}\sum_{d=0}^{\infty}\frac{q^d}{n!}G_d^\bullet(z_1,\cdots,z_n,\hbar)= \exp\left(\sum_{n=0}^{\infty}\sum_{d=0}^{\infty}\frac{q^d}{n!}G_d^\circ(z_1,\cdots,z_n,\hbar)\right).
\end{equation}
by using Lemma \ref{t2.4} to the r.h.s., we arrive at the following crucial identity
\begin{equation}\label{crucial}
G_{d}^\circ(z_1,\cdots,z_n,\hbar)=\frac{1}{d!(rd)!r^{d}\hbar^{(r+1)d+n}}\left<\alpha_1^{rd}\prod_{i=1}^{n}\mathcal{E}_0(\hbar z_i)\alpha_{-r}^{d}\right>^{\circ}.
\end{equation}
A direct consequence of which is, by Eq.~(\ref{badefh}), we have $\ln\Phi(x)$ is the principal specialization of (after setting negative $T_i$'s to 0)
\begin{equation}\label{lnphi}
\ln\Phi(x)=\sum_{n=0}^{\infty}\sum_{d=0}^{\infty}\frac{q^{rd}}{n!d!(rd)!r^{d}\hbar^{(r+1)d+n}}\left<\alpha_1^{rd}\prod_{i=1}^{n}\mathcal{E}_0(\hbar z_i)\alpha_{-r}^{d}\right>^{\circ*},
\end{equation}
where $*$ in the superscript refer to the principal specialization Eq.~(\ref{pringwtoy}). 
The reason why we want to write our formulas in the form of connected vev, is due calculational convenience. Following \cite{okounkov2006gromov} we define:
\begin{equation*}
G\left(
\begin{matrix}
a_1&\cdots&a_n\\
z_1&\cdots&z_n\\
\end{matrix}\right):=\left<\prod_{i=1}^{n}{\mathcal{E}_{a_i}(z_i)}\right>^\circ
\end{equation*}
then all $G$'s can be calculated from the following recursion relation and initial conditions
\begin{equation*}
G\left(
\begin{matrix}
a_1&\cdots&a_n\\
z_1&\cdots&z_n\\
\end{matrix}\right)=\sum_{i=2}^{n}\zeta\left(
\det
\left[
\begin{matrix}
a_1& a_i\\
z_1 &z_i \\
\end{matrix}
\right]
\right)
G\left(
\begin{matrix}
a_2&\cdots&a_i+a_1&\cdots&a_n\\
z_2&\cdots&z_i+z_1&\cdots&z_n\\
\end{matrix}
\right)
\end{equation*}
and
\begin{equation*}
\begin{split}
G\left(
\begin{matrix}
a_1&\cdots&a_n\\
z_1&\cdots&z_n\\
\end{matrix}\right)&=0,\quad a_1\leq0\\
G\left(
\begin{matrix}
0\\
z\\
\end{matrix}\right)&=\frac{1}{\zeta(z)},
\end{split}
\end{equation*}
we have immediately for the special case
\begin{equation*}
G\left(
\begin{matrix}
r&-r\\
z_1&z_2\\
\end{matrix}\right)=\left\{
\begin{split}
&0,&\quad r\leq0\\
&\frac{\zeta(r(z_1+z_2))}{\zeta(z_1+z_2)}=\frac{r\mathcal{S}(r(z1+z2))}{\mathcal{S}(z1+z2)},&\quad r>0
\end{split}\right.
\end{equation*}
where $\mathcal{S}(z) = \zeta(z)/z$. From above recursion relation, one can easily notice all the 0-point correlation functions are vanishing as we have claimed.\par
From the above properties, we can see that $\forall n\geq2$, $\left<\prod_{i=1}^{n}{\mathcal{E}_{a_i}(z_i)}\right>^\circ$ is a Taylor series in all the variables $z_i$, i.e. no negative degree terms. Therefore, the only place we get a negative order of $z_i$ in $\left<\alpha_1^{rd}\prod_{i=1}^{n}\mathcal{E}_0(\hbar z_i)\alpha_{-r}^{d}\right>^\circ$ is through the $n=1,d=0$ contribution, which is $\frac{1}{\zeta(\hbar z_1)}$ a Laurent series of $z_1$. Since the $n=1$ contribution contains no mixing of $z_i$ coordinates, therefore the order of the two operations: setting all negative $T_i$'s to 0 and taking the principal specialization, are interchangeable, i.e. we can first taking the principal specialization and then remove the non-negative degree parts of $T$. where for negative parts of $T_i$ the principal specialization is artificially defined by $T_{-2}=x-x\ln x$ and $T_{-1}=0$, which corresponding to set $\left<\tau_{-2}(\omega)\right>_{0,1,0/r}^0=1$ and $\tau_{-1}(\omega)=0$. Since the operation of removing unstable contribution is straightforward, in what follows we will keep the unstable contribution from $d=0$ in the definition of $\Phi(x)$, i.e we will define the wave function as
\begin{equation}\label{lnpsi}
\ln\Phi(x)=\frac{1}{\hbar}(x-x\ln x)+\sum_{n=0}^{\infty}\sum_{d=0}^{\infty}\frac{q^{rd}}{n!d!(rd)!r^{d}\hbar^{(r+1)d+n}}\left<\alpha_1^{rd}\prod_{i=1}^{n}\mathcal{E}_0(\hbar z_i)\alpha_{-r}^{d}\right>^{\circ*},
\end{equation}
%where we have added another $\ln x$ which comes from the definition of KP Baker-Akhiezer function Eq.~(\ref{kpprin}).
%{\color{red}I think the additional x should be considered as Singular part?}. 

As we have mentioned before, the choice of the unstable contribution is quite subtle, and the reason we choose these definitions is to accord with the result from vev.\par
In order to carry out the principal specialization we recall the proof of Corollary 4.2 in \cite{dunin2017quantum}. The key observation is the following lemma
\begin{lemma}\label{laurenta}
For arbitrary Laurent series $A(x)=\sum_{i=-1}^{\infty}a_{i}x^i$ one has
\begin{equation*}
A\left(-\hbar\partial_x\right)(\ln x) = a_{-1}\left(\frac{x-x\ln x}{\hbar}\right)+a_0\ln x-\sum_{i=1}^{\infty}a_i\frac{(i-1)!\hbar^i}{x^i}.
\end{equation*} 
\end{lemma}
\begin{proof}
Expanding the lhs and taking the corresponding power of derivatives on each term for $\ln x$.
\end{proof}

This lemma gives an explicit realization of  the principal specialization, with only a slight miss-matching, that there is a $\hbar$ missing from each insertion. Therefore a prefactor $\frac{1}{\hbar^n}$ will be cancelled out, and the homogeneity under the grading Eq.~(\ref{grading}) is now manifest.\par
Applying to Eq.~(\ref{lnpsi}), we have
\begin{equation*}
\begin{split}
\ln\Phi(x)&=\sum_{n=0}^{\infty}\sum_{d=0}^{\infty}\frac{q^{rd}}{n!d!(rd)!r^{d}\hbar^{(r+1)d}}\left<\alpha_1^{rd}\left(\mathcal{E}_0\left(-\hbar\partial_x\right)\ln x\right)^n\alpha_{-r}^{d}\right>^{\circ}\\
&=\sum_{d=0}^{\infty}\frac{q^{rd}}{d!(rd)!r^{d}\hbar^{(r+1)d}}\left<\alpha_1^{rd}e^{\mathcal{E}_0\left(-\hbar\partial_x\right)\ln x}\alpha_{-r}^{d}\right>^{\circ}.
\end{split}
\end{equation*}
The reason that we have dropped the $*$ sign in the superscript, is in the case of $\mathbb{P}[r],\,r\neq1$, there is no non-vanishing 0-point functions, and since we have added the 0-degree contributions, all the possible unstable contributions is included, thus we can safely remove this sign.\par

\begin{remark}\label{rr=1}
For the case $r=1$, the 0-point contribution will be non-trivial since
$$
\left<e^{\frac{\alpha_1}{\hbar}}q^He^{\frac{\alpha_{-1}}{\hbar}}\right> = e^{\frac{q}{\hbar^2}}.
$$
Therefore, one need to add a normalization factor $-\frac{q}{\hbar^2}$ into the above equation. This additional factor will not be manifested in the sequel, while the reader need to remember whenever the case $r=1$ is considered, there will be a normalization factor $e^{-\frac{q}{\hbar^2}}$ appear in the front of the wave function or Baker-Akhiezer function.
\end{remark}
It is not hard to see Lemma \ref{t2.4} is still valid after implementing principal specialization on both sides. Therefore exponentiate the above equation we have
\begin{equation*}
\Phi(x) = \sum_{d=0}^{\infty}\frac{q^{rd}}{d!(rd)!r^{d}\hbar^{(r+1)d}}\left<\alpha_1^{rd}e^{\mathcal{E}_0\left(-\hbar\partial_x\right)\ln x}\alpha_{-r}^{d}\right>.
\end{equation*}\par
From the definition Eq.~(\ref{defE}) of the $\mathcal{E}_0$ operator, we get
\begin{equation*}
\mathcal{E}_0\left(-\hbar\partial_x\right)\ln x = \sum_{k\in\mathbb{Z}+\frac{1}{2}}\ln(x-k\hbar):\psi_k\psi_k^*:+\frac{1}{\zeta\left(-\hbar\partial_x\right)}(\ln x):=A+\si(x).
\end{equation*}
Note for the operator on the denominator, it can only be understood if the whole function has a well-defined Taylor expansion, since $\frac{z}{\zeta(z)}= \frac{1}{\mathcal{S}(z)}:=T(z)$ has a well defined Taylor expansion and is related to the generating series $B(z)$ of Bernoulli numbers by $e^{-z/2}\frac{1}{\mathcal{S}(z)}=B(z)=\frac{z}{e^z-1}$, we rewrite the singular term as
\begin{equation}\label{defa2}
\begin{split}
\si(x)&=\frac{1}{\zeta\left(-\hbar\partial_x\right)}(\ln x) = T\left(-\hbar\partial_x\right)\left(\frac{x-x\ln x}{\hbar}\right)=\exp(\frac{-\hbar}{2}\partial_x)B\left(-\hbar\partial_x\right)\left(\frac{x-x\ln x}{\hbar}\right)\\
&=\exp(\frac{-\hbar}{2}\partial_x)\left[\frac{\left(x-x\ln x\right)}{\hbar}-\frac{1}{2}\ln x-\sum_{k=1}^{\infty}\frac{B_{2k}}{2k(2k-1)}\left(\frac{\hbar}{x}\right)^{2k-1}\right]\\
&=\frac{(x-\frac{\hbar}{2})-(x-\frac{\hbar}{2})\ln (x-\frac{\hbar}{2})}{\hbar}-\frac{1}{2}\ln (x-\frac{\hbar}{2})-\sum_{k=1}^{\infty}\frac{B_{2k}}{2k(2k-1)}\left(\frac{\hbar}{x-\frac{\hbar}{2}}\right)^{2k-1}\\
&=\frac{x-x\ln x}{\hbar}+\sum_{k=2}^{\infty}\frac{1}{k2^k}\left(\frac{\hbar}{x}\right)^{k-1}-\sum_{k=1}^{\infty}\frac{B_{2k}}{2k(2k-1)}\left(\frac{\hbar}{x-\frac{\hbar}{2}}\right)^{2k-1}.
\end{split}
\end{equation}
Now we can answer the question why we have introduced the specialization $x_{-2}=x-x\ln x$ and $x_{-1} = 0$, i.e. they are just artificially defined in order to match this calculation. \par
Due to $\si(x)$ is an ordinary function, one has naively $[A,\si(x)]=0$, therefore, we have
\begin{equation}\label{quantum2}
\Phi(x) = e^{\si(x)}\sum_{d=0}^{\infty}\frac{q^{rd}}{d!(rd)!r^{d}\hbar^{(r+1)d}}\left<\alpha_1^{rd}e^{A}\alpha_{-r}^{d}\right>.
\end{equation}
Recall the Stirling formula

\begin{lemma}[Stirling Formula]
We have the following asymptotic expansion for log-Gamma function, for $x\rightarrow\infty,|\arg(x)|<\pi$

$$
\ln\Gamma(x)\sim x\ln x-x-\frac{1}{2}\ln(\frac{x}{2\pi})+\sum_{k=1}^{\infty}\frac{B_{2k}}{2k(2k-1)}\frac{1}{x^{2k-1}}.
$$
\end{lemma}

Then we have
\begin{equation}\label{gxh}
\frac{\left(x-x\ln x\right)}{\hbar}-\frac{1}{2}\ln x-\sum_{k=1}^{\infty}\frac{B_{2k}}{2k(2k-1)}\left(\frac{\hbar}{x}\right)^{2k-1} \sim -\ln\Gamma(\frac{x}{\hbar})-\left(\frac{x}{\hbar}-\half\right)\ln\hbar-\ln x +\half\ln(2\pi),
\end{equation}
for $x\rightarrow\infty,|\arg(x)|<\pi$. Hence
\begin{proposition}\label{psingx}
$\si(x)$ is the asymptotic expansion of the function:
\begin{equation}\label{singx}
\si(x) \sim -\ln\Gamma(\frac{x}{\hbar}+\frac{1}{2})- \frac{x}{\hbar}\ln\hbar +\frac{1}{2}\ln(2\pi),
\end{equation}
for $x\rightarrow\infty,|\arg(x)|<\pi$.
\end{proposition}

\begin{proof}
Using Eq.~(\ref{gxh}), Eq.~(\ref{defa2}) together with the definition of Gamma function.
\end{proof}

%where we have include the additional $x$ from the prefactor to the singular part. The reason that the singular part is in such a form is not a coincidence, since it comes from $\zeta(z)$ in the definition of $\mathcal{E}_0(x)$, and the reason that the authors in \cite{okounkov2006gromov} include such a term is by the need of regularizing the $\Gamma(x+k)$.\par
%%
%%
Now in order to simplify the notation (avoiding frequently writing $\half$) we define
\begin{equation*}
\sigma(x) = e^{\frac{\hbar}{2}\partial_x}\sum_{d=0}^{\infty}\frac{q^{rd}}{d!(rd)!r^{d}\hbar^{(r+1)d}}\left<\alpha_1^{rd}e^{A}\alpha_{-r}^{d}\right>,
\end{equation*}
i.e. one has
\begin{equation*}
\Phi(x) = e^{\si(x)}e^{-\frac{\hbar}{2}\partial_x}\sigma(x).
\end{equation*}\par

Therefore, insert the expression of $\si(x)$ in Proposition \ref{psingx} back into Corollary \ref{rr} , one can see the only task remains is to calculate the functions 
\begin{equation}\label{phieq}
\sigma(x):=\sum_{d=0}^{\infty}\frac{q^{rd}}{d!(rd)!r^{d}\hbar^{(r+1)d}}\left<\alpha_1^{rd}\exp\left( \sum_{k\in\mathbb{Z}+\frac{1}{2}}\ln(x-(k-\frac{1}{2})\hbar):\psi_k\psi_k^*:\right)\alpha_{-r}^{d}\right>:=\sum_{d=0}^{\infty}\frac{q^{rd}X_d}{\hbar^{(r+1)d}},
\end{equation}
Note the summation is over $\half$ integer, therefore the $\half$ is cancelled out.
%which is the Baker-Akhiezer function relate to the KP hierarchy.
%{\color{red}here is an issue that I have not fully understood, the singular term actually includes not only unstable contribution, but also part of the stable contribution those comes from $\zeta(x)$ part of $\mathcal{E}_0$. However, after exclude the unstable contribution, the result is already a KP, it seems if we exclude the singular part, will go to another KP.}
%%
%%
%%%%%%%%%%%%%%%%%%%%%%%%%%%%%%%%%%%%%%%%%%%%%%%%%%%%%%%%%%%%%%%%%%%%%%%%%%%%%%%%%%%
\subsection{Wave function and Baker-Akhiezer function}
To begin with, we define $v_{\lambda}$ for the vector in Sato's infinite Grassmannian, corresponding to the partition $\lambda$, i.e.
\begin{equation*}
v_{\lambda}=\left(\lambda_1-\frac{1}{2}\right)\wedge\left(\lambda_2-\frac{3}{2}\right)\wedge \left(\lambda_3-\frac{5}{2}\right)\wedge\cdots,
\end{equation*}
and we denote $v_{\emptyset}$ for the 0 vacuum. Recall a well-known equation
\begin{equation*}
\prod_{i=1}^{l(\lambda)}\alpha_{-\lambda_{i}}v_{\emptyset} = \sum_{|\rho|=|\lambda|}\chi^{\rho}_{\lambda}v_{\rho},
\end{equation*}
where $\chi^{\rho}_{\lambda}$ is the character of the representation of symmetric group $S_{|\lambda|}$, corresponding to the partition $\rho$, evaluating on the conjugacy class corresponding to $\lambda$. The above equation is, as pointed out in \cite{okounkov2006gromov}, is equivalent to the Murnaghan–Nakayama rule, which itself is a combinatorial method to calculate the characters of symmetric group.\par
Therefore we have
\begin{equation*}
\alpha_{-r}^{d}v_{\emptyset} = \sum_{|\lambda|=rd}\chi^{\lambda}_{(r)^d}\cdot v_{\lambda},
\end{equation*}
where $(r)^d$ is the partition of $rd$ with $d$ $r$'s. For the other side we have
\begin{equation*}
\bra{v_{\emptyset}}\alpha_1^{rd}\ket{v_{\lambda}} = \dim(\lambda) = \chi^{\lambda}_{(1)^{rd}},
\end{equation*}
where $\dim(\lambda)$ is the number of standard Young tableaux of shape $\lambda$, and by $\bra{v_{\lambda}}$ we refer to the dual of $\ket{v_{\lambda}}:=v_{\lambda}$ w.r.t. the Hermitian inner product. We proceed by utilizing Lemma \ref{proeigen}
Then we have $v_{\lambda}$ is a eigenvector for the operator $\exp(A)$ with eigenvalue
\begin{equation*}
\exp\left(\sum_{i=1}^{\infty}\ln(x+(i-\lambda_i)\hbar)-\ln(x+i\hbar)\right)=\prod_{i=1}^{\infty}\frac{x+(i-\lambda_i)\hbar}{x+i\hbar}.
\end{equation*}
Therefore, we have
\begin{equation} \label{formulaofXd}
X_d=\frac{1}{d!(rd)!r^{d}}\sum_{|\lambda|=rd}\chi^{\lambda}_{(1)^{rd}}\chi^{\lambda}_{(r)^d} \prod_{i=1}^{\infty}\frac{x+(i-\lambda_i)\hbar}{x+i\hbar}.
\end{equation}

The above formula can be further simplified. Finally we will prove
\begin{proposition}\label{p2.1}
The following closed formula of $X_d$ holds:
\begin{equation*}
X_d = \frac{(-1)^d}{r^dd!}\prod_{i=1}^{rd}\frac{\hbar}{x+i\hbar}.
\end{equation*}
\end{proposition}
\begin{proof}
By Equation~\eqref{formulaofXd}, this proposition follows from   the following key idendity
\begin{equation} \label{Xdidentity}
X_d :=\frac{1}{d!(rd)!r^{d}}\sum_{|\lambda|=rd}\chi^{\lambda}_{(1)^{rd}}\chi^{\lambda}_{(r)^d} \prod_{i=1}^{\infty}\frac{x+(i-\lambda_i)\hbar}{x+i\hbar}
= \frac{(-1)^d}{r^dd!}\prod_{i=1}^{rd}\frac{\hbar}{x+i\hbar}:=\tilde{X}_d.
\end{equation}
which will be proved in Lemma \ref{lemA1}, Appendix II.
\end{proof}\par
Therefore, we arrive at
\begin{theorem2.1'}[The $t=0$ part of Theorem  \ref{onept}]
%[First half part of Theorem \ref{onept1}]
The wave function for the orbifold Gromov-Witten theory of $\mathbb{P}[r]$ with $t=0$ has the following closed form
\begin{equation*}
\begin{split}
\Phi(0,x)&= e^{\si(x)}\sum_{d=0}^{\infty}\frac{(-1)^dq^{rd}}{r^dd!\hbar^{(r+1)d}}\prod_{i=0}^{rd}\frac{\hbar}{x+(i-\frac{1}{2})\hbar}:=e^{\si(x)}e^{-\frac{\hbar}{2}\partial_x}\sigma(x)\\
& = (2\pi)^{\half}\sum_{d=0}^{\infty}\frac{(-1)^dq^{rd}}{r^dd!}\frac{\hbar^{-\frac{x}{\hbar}-(r+1)d}}{\Gamma(\frac{x}{\hbar}+rd+\frac{1}{2})}: = \sum_{d=0}^{\infty}q^{rd}\Phi^d(x).
\end{split}
\end{equation*}
For the case $r=1$ there will be an additional normalization factor $e^{-\frac{q}{\hbar^2}}$.
\end{theorem2.1'}

\begin{proof}
Applying Proposition \ref{p2.1} together with Eq.~(\ref{phieq}) gives the expression of $\Phi(0,x)$. The additional normalization factor $e^{-\frac{q}{\hbar^2}}$ comes from Remarks \ref{rr=1}.
\end{proof}
For Baker-Akhiezer function we have

$$
\Psi(x) =(2\pi)^{\half}e^{\frac{x}{\hbar}\ln x - \frac{x}{\hbar}}\sum_{d=0}^{\infty}\frac{(-1)^dq^{rd}}{r^dd!}\frac{\hbar^{-\frac{x}{\hbar}-(r+1)d}}{\Gamma(\frac{x}{\hbar}+rd+\half)}.
$$

\subsection{Quantum curve equation}
%{\Chen: I will keep the content before here unchanged, in case you need to change any notations or something, and I will edit the contents below offline first.}
As we mentioned in the Introduction, for any Baker-Akhiezer function, there should exist an operator annihilate it, which can be viewed as the quantization of the classical spectral curve. In this subsection we will prove

\begin{theorem2.2'}[The $t=0$ part of Theorem \ref{t2.2}]
The wave function of the Gromov-Witten theory of $\mathbb{P}[r]$ satisfy the following quantum curve equation
$$
\hat{H}_0(x,-\hbar\partial_x)\cdot\Phi(x) = 0.
$$
where
\begin{equation*}
%\begin{split}
\hat{H}_0(x,-\hbar\partial_x) = e^{-\hbar\frac{\partial}{\partial x}}+q^re^{r\hbar\frac{\partial}{\partial x}}-x+\frac{\hbar}{2}.
%&\hat{H}^*(y,\hbar\partial_y) = e^{\hbar\frac{\partial}{\partial y}}+q^re^{t\hbar rd}e^{-r\hbar\frac{\partial}{\partial x}}-y-\frac{\hbar}{2}.
%\end{split}
\end{equation*}
\end{theorem2.2'}\par
\begin{proof}
%First we notice
%\begin{equation*}
%\Psi(t,x) = e^{\left(t-\frac{\hbar}{2}\right)\partial_x}\Psi(\frac{\hbar}{2},x),
%\end{equation*}
%which gives
%\begin{equation*}
%\hat H_t(x, \hbar\partial_x)  \Psi(t,x) = \hat H_t(x, \hbar\partial_x) e^{\left(t-\frac{\hbar}{2}\right)\partial_x}\Psi(\frac{\hbar}{2},x)= e^{\left(t-\frac{\hbar}{2}\right)\partial_x} \hat H_{\frac{\hbar}{2}}(x, \hbar\partial_x)\Psi(\frac{\hbar}{2},x)  
%\end{equation*}
%i.e. we only need to prove for case where $t = \frac{\hbar}{2}$. 
By Remarks \ref{rrr}, we will only prove the theorem for the case $t\neq0$. Recall the expression for $\Phi(x)$ from Theorem \ref{onept}, which means we need to prove
\begin{equation}\label{qchalf}
\left(e^{-\hbar \partial_x} + q^r\, e^{r\hbar \partial_x} -x+\frac{\hbar}{2}\right)\left((2\pi)^{\half}\sum_{d=0}^{\infty}\frac{(-1)^dq^{rd}}{r^dd!}\frac{\hbar^{-\frac{x}{\hbar}-(r+1)d}}{\Gamma(\frac{x}{\hbar}+rd+\half)}\right)=0,
\end{equation}
and it is a direct consequence of
\begin{equation*}
\left(e^{-\hbar \partial_x}-x+\frac{\hbar}{2}\right) \Psi^d(x) + e^{r\hbar \partial_x}\Psi^{d-1}(x) = 0,\quad \forall d\geq1,
\end{equation*}
since $\Psi^0 = 1$, summing over $d$ with weight $q^{rd}$ gives precisely Eq.~(\ref{qchalf}).
\end{proof}
%%%%%%%%%%%%%%%%%%% %%%%%%%%%%%%%%%%%%%%%%%%%%%%%%%%%%%%%%%%%%%%%%%%%%%%%%%%%%%%%%%%
\section{The Fermionic two Point Function}\label{s3}
In the previous section, we already calculated the Baker-Akhiezer function of the KP hierarchy, whose specialization at $t=0$ is the first basis of $V$, the corresponding point in Sato's infinite Grassmannian. However, in order to fully determine the solution of the KP hierarchy, i.e. derive the Bogoliubov transformation of the vacuum, one has to derive the remaining basis vectors, which will be our major focus for the remaining of this section.\par
In this section, we will first derive an admissible basis for the KP hierarchy, and further derive the canonical basis from it by a kind of orthonormalization. Finally, assemble the canonical basis properly will give us the canonical Fermionic bilinear, thus the Bogoliubov transformation.\par
%%%%%%%%%%%%%%%%%%%%%%%%%%%%%%%%%%%%%%%%%%%%%%%%%%%%%%%%%%%%%%%%%%%%%%%%%%%%%%%%%%%
\subsection{Baker-Akhiezer function}\label{s3.1}
The first key ingredient is the Baker-Akhiezer function with arbitrary $t$ (this is the reason why we postpone the $t\neq0$ case in last section to here). The Baker-Akhiezer function defined in Eq.~(\ref{BAfunc}) has the following nice property
\begin{proposition}\label{pba}
The Baker-Akhiezer function $\Phi(t,x)$ locates inside the semi-infinite vector space $V$, for $\forall t\in\C$. Moreover, one has
\begin{equation*}
\partial^n_t \Psi(t,x)\big|_{t=t_0} \in V,\quad \forall t_0\in\C,\, n\in\mathbb{Z}_+.
\end{equation*}
\end{proposition}
\begin{proof}
This is a result from\cite{segal1985loop}, see also, \cite{alexandrov2015enumerative,babelon2003introduction}.
\end{proof}
To calculate the Baker-Akhiezer function with arbitrary $t$, we need the following non-equivariant divisor equation in operator formalism, which specialized from the equivariant one, i.e. Proposition VI.1. of \cite{johnson2009equivariant}
\begin{proposition}\label{pd1}[Divisor Equation]
\begin{equation*}
[z_0^1]G_{rd}^\bullet(z_0,z_1\cdots,z_n,\hbar) = (rd-\frac{1}{24})G_{rd}^\bullet(z_1,z_2\cdots,z_n,\hbar).
\end{equation*}
\end{proposition}
\begin{proof}
Simply specialize the Proposition VI.1. of \cite{johnson2009equivariant} by taking $m=0$, neglecting $\mathfrak{r}$, which is trivial for our case, and taking the non-equivariant limit.
\end{proof}
Now let us calculate the Baker-Akhiezer function with $t\neq0$. In what follows we will drop the evaluation symbol $\big|_{t=t_0}$, and regard $t=t_0$ by a slight misuse of notation. Since $[\alpha_1,\psi(x)] = x\psi(x)$, and $\bra{1}\psi(x) = \bra{0} \Gamma_{+}(-\{x^{-1}\})$, we have
$$
\Psi(t,x) = e^{tx}\frac{\left<  \Gamma_{+}(-\{x^{-1}\}) e^{t\alpha_1}  |V\right>}{\braket{0|e^{t\alpha_1}|V}}.
$$
where $\{x^{-1}\}$ refers to the principal specialization $t_k = -\frac{1}{kx^k}$.

Recall $\ket{V} = W^{-1}e^{\frac{\alpha_1}{\hbar}}q^{H}e^{\frac{\alpha_{-r}}{r\hbar}}\ket{0}$ and
$$
W^{-1}\mathbf A[0]W = \frac{\alpha_{1}}{\hbar}, \, W^{-1}e^{\sum_i T_i \mathbf A[i]}W = e^{\sum_{i}\frac{T_i}{\hbar}\alpha_{i+1}},
$$
we have
\begin{equation*}
\begin{split}
\Psi(t,x) &= \frac{\left<  \Gamma_{+}(-\{x^{-1}\}) e^{t(\alpha_1+x)}  |V\right>}{\braket{0|e^{t\alpha_1}|V}} =  \frac{\left< \Gamma_{+}(-\{x^{-1}\}) e^{t(\alpha_1+x)} W^{-1}e^{\frac{\alpha_1}{\hbar}}q^{H}e^{\frac{\alpha_{-r}}{r\hbar}}\right>}{\braket{0|e^{t\alpha_1}|V}}\\
& = e^{tx}\frac{\left< \Gamma_{+}(-\{x^{-1}\}) W^{-1} e^{t\hbar\mathbf A[0]} e^{\frac{\alpha_1}{\hbar}} q^{H}e^{\frac{\alpha_{-r}}{r\hbar}}\right>}{\braket{0|e^{t\alpha_1}|V}}.
\end{split}
\end{equation*}
Since dressing op $W$ fixes the covacuum, we have
$$
\Psi(t,x) = e^{tx}\frac{\left< e^{t\hbar\mathbf A[0]}e^{\sum_i T_i\mathbf A[i]} e^{\frac{\alpha_1}{\hbar}} q^{H}e^{\frac{\alpha_{-r}}{r\hbar}}\right>^*}{\braket{0|e^{t\alpha_1}|V}}.
$$
After adding the unstable contribution by multipltiplication: $\Phi(t,x):=e^{(x-x\ln x-t\hbar x)/\hbar}\Psi(t,x)$, we have for wave function (Baker-Akhiezer function with unstable term included)
$$
\Phi(t,x)  =\frac{\left<e^{t\hbar\mathbf A[0]}e^{\mathbf A(-\hbar\partial_x)\ln x}q^He^{\frac{\alpha_{-r}}{r\hbar}}\right>}{\braket{0|e^{t\alpha_1}|V}} = \sum_{d=0}^{\infty}\frac{q^{rd}}{r^dd!\hbar^{d}}
\frac{\left< e^{t\hbar\mathbf A[0]}e^{\mathbf A(-\hbar\partial_x)\ln x} e^{\alpha_1}\alpha_{-r}^d \right>}{\braket{0|e^{t\alpha_1}|V}}.
$$
By divisor equation Proposition \ref{pd1}
\begin{equation*}
\Phi(t,x) =  \sum_{d=0}^{\infty}\frac{q^{rd}}{r^dd!\hbar^{d}}e^{t\hbar\left(rd-\frac{1}{24}\right)}\frac{\left< e^{\mathbf A(-\hbar\partial_x)\ln x} e^{\alpha_1}\alpha_{-r}^d \right>}{\braket{0|e^{t\alpha_1}|V}} ,
\end{equation*}
Recall
$$
e^{\alpha_1} \mathcal{E}_0(z) e^{-\alpha_{1}}  = \mathbf A(z,\hbar),
$$
The overall constant factor $e^{-t\frac{\hbar}{24}}$, will be cancelled out by the normalizing factor:

\begin{lemma}
$$
\braket{0|e^{t\alpha_1}|V} = e^{-\frac{t\hbar}{24}}.
$$
\end{lemma}
\begin{proof}
Recall the expression of $\ket{V}$ and divisor equation Proposition \ref{pd1}, we have
$$
\braket{0|e^{t\alpha_1}|V} = \sum_{d=0}^{\infty}\frac{q^{rd}}{r^dd!(rd)!\hbar^{d}}e^{t\hbar\left(rd-\frac{1}{24}\right)}\left<\alpha_1^{rd}\alpha_{-r}^d \right>.
$$
For the case $r\neq 1$, by orthogonality, we  arrive at
$$
\braket{0|e^{t\alpha_1}|V} = e^{-\frac{t\hbar}{24}}.
$$

For the case $r =  1$, recall there should be an additional normalization factor $e^{-\frac{qe^{t\hbar}}{\hbar^2}}$ in front of the Baker-Akhiezer function, and we have

$$
\braket{0|e^{t\alpha_1}|V} = \sum_{d=0}^{\infty}\frac{q^{d}}{(d!)^2\hbar^{d}}e^{t\hbar\left(rd-\frac{1}{24}\right)}\left<\alpha_1^{d}\alpha_{-1}^d \right> =e^{\frac{qe^{t\hbar}}{\hbar^2}}\cdot e^{-\frac{t\hbar}{24}}
$$

The first factor serves as the additional normalization factor as in Remarks \ref{rr=1}, which is the case for $t=0$ , thus we have completed the prove.
\end{proof}

Now we arrive at
$$
\Phi(t,x) =  \sum_{d=0}^{\infty}\frac{q^{rd}}{r^dd!(rd)!\hbar^{(r+1)d}}e^{t\hbar rd}\left< \alpha_1^{rd}e^{ \mathcal{E}_0(-\hbar\partial_x)\ln x}\alpha_{-r}^d \right> ,
$$
since the vev in the above expression is independent from $t$, which means we have the following $t$-evolution formula
\begin{lemma}\label{psiphi}[ $t$-evolution of wave function]
\begin{equation*}
\Phi(t,x) =  e^{t\hbar q\partial_q}\Phi(0,x).
\end{equation*}
\end{lemma}
Recall the expression for the $t=0$ wave function in Theorem 2.1'
$$
\Phi(0,x) =(2\pi)^{\half}\sum_{d=0}^{\infty}\frac{(-1)^{d}q^{rd}}{r^dd!}\frac{\hbar^{-\frac{x}{\hbar}-(r+1)d}}{\Gamma(\frac{x}{\hbar}+rd+\half)} 
$$
Lemma. \ref{psiphi} gives
$$
\Phi(t,x) =(2\pi)^{\half}\sum_{d=0}^{\infty}\frac{(-1)^{d}q^{rd}e^{t\hbar rd}}{r^dd!}\frac{\hbar^{-\frac{x}{\hbar}-(r+1)d}}{\Gamma(\frac{x}{\hbar}+rd+\half)} 
$$

Thus we have the following expression for $t\neq0$ Baker-Akhiezer function
$$
\Psi(t,x) =(2\pi)^{\half}e^{(x\ln x - x+t\hbar x)/\hbar}\sum_{d=0}^{\infty}\frac{(-1)^{d}q^{rd}e^{t\hbar rd}}{r^dd!}\frac{\hbar^{-\frac{x}{\hbar}-(r+1)d}}{\Gamma(\frac{x}{\hbar}+rd+\half)} 
$$
The Baker-Akhiezer function satisfies the following property
\begin{lemma}\label{ll}
\begin{equation*}
\partial_t \Psi(t,x) = e^{(x\ln x - x+t\hbar x)/\hbar}e^{-\hbar\partial_x}e^{(-x\ln x + x-t\hbar x)/\hbar}\Psi(t,x),\quad\forall t\in\C.
\end{equation*}.
\end{lemma}
\begin{proof}
By definition we have
\begin{equation*}
\begin{split}
\partial_t \Psi(t,x)& = (2\pi)^{\half}e^{(x\ln x - x+t\hbar x)/\hbar}\sum_{d=1}^{\infty}\frac{(-1)^{d}q^{rd}e^{t\hbar rd}(x+\hbar rd)}{r^dd!}\frac{\hbar^{-\frac{x}{\hbar}-(r+1)d}}{\Gamma(\frac{x}{\hbar}+rd+\half)}\\
& =e^{(x\ln x - x+t\hbar x)/\hbar}\left(e^{-\hbar\partial_x}-\frac{\hbar}{2}\right)e^{(-x\ln x + x-t\hbar x)/\hbar}\Psi(t,x)
\end{split}
\end{equation*}.
\end{proof}

\begin{remark}
Now one can see Theorem 2.1' and Theorem 2.2' can be naturally generalized to the case $t\neq0$ by using the evolution Lemma \ref{psiphi}.
\end{remark}
%%%%%%%%%%%%%%%%%%%%%%%%%%%%%%%%%%%%%%%%%%%%%%%%%%%%%%%%%%%%%%%%%%%%%%%%%%%%%%%%%%%
\subsection{Admissible basis and The lifting operator }
\begin{definition}
A set of admissible basis for $\ket{V}$, is a set of vectors $\{\phi_k\},\,k\in\mathbb N$ with
$$
\phi_k = x^k+\sum_{i=1-k}^{\infty} b_ix^{-i},\quad\forall k\in\mathbb N,
$$
and $\ket{V} =\underline{\phi_0}\wedge
\underline{\phi_1}\wedge\underline{\phi_2}\wedge\cdots $
\end{definition}

Apparently, the admissible basis for any given $\ket{V}$ is not unique. Then we can define the following important operator

\begin{definition}
A lifting operator $D_x$ for a KP hierarchy corresponding to $\ket{V}=\underline{\phi_0}\wedge
\underline{\phi_1}\wedge\underline{\phi_2}\wedge\cdots$ in Sato's semi-infinite Grassmannian, is a KS operator, that satisfies the following property
\begin{equation}\label{dkmphi}
D_x^k\phi_m = \mathrm{span}\{\phi_0,\phi_1,\cdots,\phi_k+m\},\quad\forall k,m\in\mathbb{N}.
\end{equation}
for arbitrary admissible basis $\{\phi_k\}$, and we have

$$
D_x^k\phi_m = x^{k+m}+ O(x^{k+m-1}).
$$
\end{definition}
Note the requirement for $D_x$ to be a KS operator is included in property Eq.~(\ref{dkmphi}). In our case, we have
\begin{proposition}
The following operator is a lifting operator for KP hierarchy $V$
$$
D_x := e^{\frac{x}{\hbar}\ln x - \frac{x}{\hbar}}\left(e^{-\hbar\partial_x}\right)e^{-\frac{x}{\hbar}\ln x + \frac{x}{\hbar}}
$$
\end{proposition}
\begin{proof}
Since Lemma \ref{ll} is valid for all $t$,  setting $t=0$ gives
$$
\partial_t \Psi(t,x)\big|_{t = 0}=e^{\frac{x}{\hbar}\ln x - \frac{x}{\hbar}}\left(e^{-\hbar\partial_x}-\frac{\hbar}{2}\right)e^{-\frac{x}{\hbar}\ln x + \frac{x}{\hbar}}\Psi(0,x) = \left(D_x-\frac{\hbar}{2}\right)\Psi(0,x).
$$
Furthermore, by Proposition \ref{pba}, we have
$$
\left(D_x-\frac{\hbar}{2}\right)^k\Psi(0,x) = \left(D_x-\frac{\hbar}{2}\right)^k\phi_0 \in V.
$$
Moreover, by direct calculation we have $D_x$ increase positive degree of any series in $\C[t]\otimes\C[[t^{-1}]]$ by one, which means
$$
D_x\phi_0 = (\partial_t+\frac{\hbar}{2})\phi_0 \in \mathrm{span}\{\phi_0,\phi_1\}.
$$
Recursively, we have
$$
D_x^k\phi_m = D_x^k\left(x^{m}+O(x^{m-1})\right)  = x^{k}+O(x^{k-1})\in\mathrm{span}\{\phi_0,\phi_1,\cdots,\phi_k+m\},\quad\forall k,m\in\mathbb{N}.
$$

i.e. $D_x$ is an lifting operator .
\end{proof}
\begin{remark}
Proposition \ref{pba} does not enforce $\partial_t$ to be a KS operator, actually, it may only take the subspace $V' = \mathrm{span}\{\phi_0,\partial_t\phi_0,\cdots\}\subset V$ as its invariant subspace. However, when $\partial_t$ is a lifting operator, we have $V =  \mathrm{span}\{\phi_0,\partial_t\phi_0,\cdots\}$ i.e. $V'=V$, it is no doubt a KS operator.
\end{remark}
Thus an admissible basis for $V$ can be easily constructed via lifting operator $D_x$:
%, actually we will just define $\phi_k:=D_x^{k}\phi_0$, we have
\begin{proposition}
We have an admissible basis $\phi_k:=D_x^{k}\phi_0$ for $V$, with the following explicit closed formula
$$
\phi_k =(2\pi)^{\half}e^{\frac{x}{\hbar}\ln x - \frac{x}{\hbar}}\sum_{d=0}^{\infty}\frac{(-1)^dq^{rd}}{r^dd!}\frac{\hbar^{-\frac{x}{\hbar}-(r+1)d+k}}{\Gamma(\frac{x}{\hbar}+rd+\half-k)}
$$
\end{proposition}
\begin{proof}
	Note that $\phi_0 = \Psi(0,x)$.
By Theorem \ref{onept}, we have
\begin{align*}
\phi_k(x)  = D^{k}_x\phi_0(x)=& \ (2\pi)^{\half}D^{k}_x\sum_{d=0}^{\infty}\frac{(-1)^dq^{rd}}{r^dd!}\frac{\hbar^{-\frac{x}{\hbar}-(r+1)d}}{\Gamma(\frac{x}{\hbar}+rd+\half)}\\
= & \ (2\pi)^{\half}\sum_{d=0}^{\infty}\frac{(-1)^dq^{rd}}{r^dd!}\frac{\hbar^{-\frac{x}{\hbar}-(r+1)d+k}}{\Gamma(\frac{x}{\hbar}+rd+\half-k)}.
\end{align*}
\end{proof}

For future convenience, by using Corrolary \ref{rr} we can rewrite the above expression as
\begin{equation*}
\phi_k(x) =  \rho(x) \sum_{d=0}^{\infty}\frac{(-1)^dq^{rd}}{d!r^d\hbar^{(r+1)d}}x_{[rd-k]},
\end{equation*}
%Here recall $k\in \mathbb Z$, $x_{[k]},y_{[k]}$ are defined in \eqref{defforxk} and $\rho(x)$ is defined in \eqref{defforrho}.
where for $k\in \mathbb Z$, we define $x_{[k]}, y_{[k]}$ as power series
\begin{equation}\label{defforxk}
x_{[k]}:= \frac{\Gamma(\frac{x}{\hbar}+\frac{1}{2})}{\Gamma(\frac{x}{\hbar}+i+\frac{1}{2})} \in \mathbb Q[[\hbar x^{-1}]],\quad
y_{[k]}:= \frac{\Gamma(\frac{y}{\hbar}-i+\frac{1}{2})}{\Gamma(\frac{y}{\hbar}+\frac{1}{2})}\in \mathbb Q[[\hbar y^{-1}]]
\end{equation}
% where $x_{[k]}: = \prod_{i=1}^{k}\frac{\hbar}{x+(i-\half)\hbar}, k>0$, $x_{[k]} = \prod_{i=0}^{-k-1}\frac{x-(i+\half)\hbar}{\hbar}, k<0$ and $y_{[k]}  = \prod_{i=1}^{k}\frac{\hbar}{y-(i-\half)\hbar},  k>0$, $y_{[k]} = \prod_{i=0}^{-k-1}\frac{y+(i+\half)\hbar}{\hbar},  k<0$. We also make the convention $x_{[0]}=y_{[0]}:=1$. 
and we define $\rho(x)$ as power series 
\begin{equation}\label{defforrho}
\rho(x) =(2\pi)^{\half}\frac{\hbar^{-\frac{x}{\hbar}}}{\Gamma(\frac{x}{\hbar}+\frac{1}{2})}e^{\frac{x}{\hbar}\ln x-\frac{x}{\hbar}}  \in \mathbb Q[[\hbar x^{-1}]]
\end{equation}
%$\rho(x) :=e^{\frac{x}{\hbar}\ln x-\frac{x}{\hbar}+\si(x)} $.% and $\frac{1}{\rho(x)}$ is now in $\C[[x^{-1}]]$. 

%%%%%%%%%%%%%%%%%%%%%%%%%%%%%%%%%%%%%%%%%%%%%%%%%%%%%%%%%%%%%%%%%%%%%%%%%%%%%%%%%%%
\subsection{Canonical basis and bilinear Fermionic form} \label{canonicalbasis}
\begin{definition}
The canonical basis for $\ket{V}$, is a set of vectors $\{\tilde\phi_k\},\,k\in\mathbb N$ with
$$
\tilde\phi_k = x^k+\sum_{j=1}^{\infty}B^{+-}_{-k+\half,j-\half}x^{-j},\quad\forall k\in\mathbb N,
$$
and $\ket{V} =\underline{\tilde\phi_0}\wedge
\underline{\tilde\phi_1}\wedge\underline{\tilde\phi_2}\wedge\cdots $
\end{definition}

One can easily see such a basis is uniquely defined. The canonical basis is naturally linked to the Fermionic two point function, which defined as

$$
B(x,y) : =\frac{\braket{0|\psi^*(x)\psi(y)|V}}{\braket{0|V}}.
$$

And the canonical basis is \cite{alexandrov2015enumerative}:

$$
\tilde\phi_i  = x^{i} - \frac{\braket{0|\psi(x)\psi^*_{k}|V}}{\braket{0|V}} = x^{i}+\sum_{j=0}^{\infty}B^{+-}_{-i+\half,j-\half}x^{-j}, \quad\forall i\in\mathbb N
$$

where
$$
B^{+-}_{-i+\half,j-\half} := [ x^{-i} y^{-j} ]B(x,y).
$$

However, if one approaches like this,  he/she has to solve the two point function, which is not an easy task usually. Another approach, is start from an arbitrary admissible basis.\par

As we have mentioned, the canonical basis can be calculated from any admissible basis, by a kind of \textit{orthonormalization}. Which comes from the following Lemma

\begin{lemma}\label{puiu}
Any two sets of admissible basis $\{\phi^1_k\},\,\{\phi^2_k\}$ differs from each other by an right multiplication of upper unitriangular matrix $M^{1,2}$
$$
(\cdots0,0,\phi^1_0,\phi^1_1,\phi^1_2,\cdots)\cdot M^{1,2}= (\cdots0,0,\phi^2_0,\phi^2_1,\phi^2_2,\cdots)
$$
\end{lemma}
\begin{proof}

By the definition of admissible basis, we have

$$
\mathrm{Span}\{\phi^1_0,\phi^1_1,\cdots,\phi^1_k\} = \mathrm{Span}\{\phi^2_0,\phi^2_1,\cdots,\phi^2_k\}, \quad\forall k\in\mathbb N.
$$
which equivalent to say there is a  upper unitriangular matrix $M^{1,2}$ connects them by right multiplication.
\end{proof}
In addition, by Lemma \ref{puiu}, one can see, right multiplying any upper unitriangular matrix will not change the $\ket{V}$. Now, by orthonormalization, we actually means there is a transformation matrix between any admissible basis  and canonical basis.\par

To do so, following from the definition we arrange the admissible basis vectors into the following infinite dimension matrix $A_{\infty\times\infty+}$, where $\infty+$ means semi-infinite, i.e. right half of an infinite dimensional matrix.
\begin{equation*}
A:= (A^{-+T},A^{--T})^T:=  (\phi_0,\phi_1,\phi_2\cdots),
\end{equation*}
where the coefficients of $\phi_k$ is labeled upward from the bottom, while the degree in $x$ decreasing. More precisely, 
$$
A^{--}_{i,j} = [x^{i-1}]\phi_{j-1}(x),\quad j\geq i\geq0, i,j\in\mathbb Z .
$$

Therefore, we have
\begin{equation}\label{Aminus}
A^{--}_{i,j} = [x^{i-1}]\rho(x)\sum_{d=0}^{\infty}\frac{(-1)^dq^{rd}}{r^dd!\hbar^{(r+1)d}}x_{[rd-j+1]} ,
%= \rho(x)\sum_{d=0}^{[(j-1)/r]}\frac{(-1)^dq^{rd}}{r^dd!\hbar^{(r+1)d}}(-1)^{j-i}e^{i-1}_{j-1-rd}
\end{equation}
%where for $i\leq n+1$
%\begin{equation*}
%e^{i-1}_{n} = \mathrm{e}_{i-1}(0,1,2,\cdots,n-1)
%\end{equation*}
%and $e^{i-1}_{n} := 0 $ for $i>n+1$. 
%We also have
and
\begin{equation}\label{Aplus}
A^{+-}_{-i,j} = [x^{-i}]\rho(x)\sum_{d=0}^{\infty}\frac{(-1)^dq^{rd}}{r^dd!\hbar^{(r+1)d}}x_{[rd-j+1]} = \sum_{rd>j-1}\frac{(-1)^dq^{rd}}{r^dd!\hbar^{(r+1)d}}[x^{-i}]\rho(x)x_{[rd-j+1]}.
\end{equation}
Here we recall $\rho(x)$ and $x_{[k]}$ are defined in \eqref{defforrho} and \eqref{defforxk}.

Similarly, for canonical basis $\tilde{\phi}_k$, one can define the matrix $B=(B^{-+T},\Id^{++T})^T$, $B = (\tilde{\phi}_0,\tilde{\phi}_1,\cdots)$. Then the orthonormalization can be realized via
\begin{equation}\label{atob}
B = (\tilde{\phi}_0,\tilde{\phi}_1,\cdots) = (\phi_0,\phi_1,\cdots)\cdot\left(A^{--}\right)^{-1}.
\end{equation}
which has a closed form
\begin{theorem}\label{pB}
$B=(B^{+-T},\Id^{--})^T$, where
for $k\in \mathbb Z$ $x_{[k]}, y_{[k]}$ are defined in \eqref{defforxk} as formal series in $x^{-1}$ and $y^{-1}$,
and 
$\rho$ is defined in \eqref{defforrho} as formal series in $x^{-1}$.
\end{theorem}
\begin{proof}
By Eq.~(\ref{atob}), this theorem follows from the following combinatorial identity
	\begin{equation} \label{AAB}
A^{+-}  = B^{+-}A^{--}.
\end{equation}
We will prove this identity in the Appendix III.
\end{proof}
%%%%%%%%%%%%%%%%%%%%%%%%%%%%%%%%APPENDIXAPPENDIXAPPENDIXAPPENDIXAPPENDIX%%%%%%%%%%%%%%%%%%%%%%%%%%%%%%%%%%%%%%%%%%%%%%%%%%%

\begin{appendix}

\section{Proof of several  combinatorial lemmas}
\renewcommand\thesubsection{\Roman{subsection}}
\setcounter{subsection}{0}

\subsection{Proof of Lemma \ref{t2.4}}
The following coordinates transformation will be used in this subsection for convenience,

\begin{definition}
By $\Delta$-transformation, we will refer to the following coordinate change: For $f\in\mathbb{C}[z_1,z_1^{-1},z_2,z_2^{-1}\cdots]$ and $f = z_{i_1}^{-n_1}z_{i_2}^{-n_2}\cdots z_{i_k}^{-n_k}$ as
\begin{equation*}
f^\Delta = (z_{i_1}^{n_1}z_{i_2}^{n_2}\cdots z_{i_k}^{n_k})^{\Delta} = x_{n_1-1}x_{n_1-1}\cdots x_{n_k-1},\quad i_m\in\mathbb{Z},
\end{equation*}
which can be linearly extended to the whole $\mathbb{C}[z_1,z_1^{-1},z_2,z_2^{-1}\cdots]$. Whenever we have a function in $z_i$ and superscript $\Delta$, the variable of which will be understood as $T_i$ but not $z_i$.
\end{definition}
Such an operation will forget the information of the position of the marked points and recollect the contribution of the same degree in $\psi$ class. \par

Now we will proof for $r\in\mathbb{Z}_+$, we have

\begin{lemma2.1}
The pseudo generating function $G(z,\hbar)$ has the following form for its $
\Delta$-transformation
\begin{equation*}
\begin{split}
G(z,\hbar)^\Delta &= \sum_{n=0}^{\infty}\sum_{d=0}^{\infty}\frac{q^{rd}}{n!}G_{rd}^\bullet(z_1,\cdots,z_n,\hbar)^\Delta \\
&=\sum_{n=0}^{\infty}\sum_{d=0}^{\infty}\sum_{P^*\in\mathrm{Part}^{*}_d[n]}\frac{1}{|\mathrm{Aut(P^*)}|}\prod_{i=1}^{l(P^*)}\frac{q^{r_id_i}}{n_i!d_i!(rd_i)!r^{d_i}\hbar^{(r+1)d_i+n_i}}\left<\alpha_1^{rd_i}\prod_{i=1}^{n_i}\mathcal{E}_0(\hbar z_{P^*_{i,1}})\alpha_{-r}^{d_i}\right>^{\circ\Delta}\\
&=\exp\left(\sum_{n=0}^{\infty}\sum_{d=0}^{\infty}\frac{q^{rd}}{n!d!(rd)!r^{d}\hbar^{(r+1)d+n}}\left<\alpha_1^{rd}\prod_{i=1}^{n}\mathcal{E}_0(\hbar z_i)\alpha_{-r}^{d}\right>^{\circ\Delta}\right),
\end{split}
\end{equation*}
where $P^*\in\mathrm{Part}^*_d[n]$ includes the following data
\begin{equation*}
P^*=\{(d_1,n_1),\cdots,(d_{l(P)},n_{l(P)})\},
\end{equation*}
\end{lemma2.1}\par
\begin{proof}
The disconnected vev $G_d^\bullet(z_1,\cdots,z_n,\hbar)$ is defined via
\begin{equation}\label{recursiveG}
G_d^\bullet(z_1,\cdots,z_n,\hbar)= \sum_{P\in\mathrm{Part}_d[n]}\frac{1}{|\mathrm{Aut}(P)|}\prod_{i=1}^{l(P)}G_{d_i}^{\circ}(z_{P_i},\hbar),
\end{equation}
We note the partition $P$ consists of the following data
\begin{equation*}
\{(d_1,P_1),\cdots,(d_{l(P)},P_{l(P)})\}.
\end{equation*}

Normally, if all $P_{i}=\{z_{i,1},\cdots,z_{i,l(P)}\}$ are not empty, then the automorphism is trivial. However, we are entitled to allow the unstable parts, therefore the 0-pint function is included, which is the only way that the automorphism is non-trivial. Note in the case of $\PP[r]$, we have, i.e.
$$
G_{rd}^\circ()=\frac{\delta_{d,1}\delta_{r,1}}{\hbar^{1+r}},\quad\forall d\geq0,
$$

which is generalized from those in OP 2006.Then we have $|\mathrm{Aut}(P)| = \frac{1}{m!}$. The next observation we need is

\begin{equation*}
G_{d}^{\circ}(z_{a_1},\cdots,z_{a_n},\hbar)^\Delta = G_{d}^{\circ}(z_{b_1},\cdots,z_{b_n},\hbar)^\Delta,\quad\forall z_1 = (z_{a_1},\cdots,z_{a_n}),z_2 = (z_{b_1},\cdots,z_{b_n}),
\end{equation*}

i.e. the position of the marked point does not matters. The same properties hold for connected functions. Therefore, for $\Delta$-transformed function, the partition $P$ is degenerating to $P^*\in\mathrm{Part}^*_d[n]$, which including the following data
\begin{equation*}
P^*=\{(d_1,n_1),\cdots,(d_{l(P)},n_{l(P)})\},
\end{equation*}

where $n_i$ is the number of the marking points, combining the parts that are identical after quantization will give an extra prefactor, which, together with the automorphism, is
\begin{equation*}
\frac{1}{|\mathrm{Aut}(P^*)|}\left(
\begin{aligned}
&n\\
n_1,n_2,&\cdots,n_{l(P)}
\end{aligned}
\right)
=\frac{1}{|\mathrm{Aut}(P^*)|}\frac{n!}{n_1!n_2!\cdots n_{l(P)}!},
\end{equation*}
where $\mathrm{Aut}(P^*)$ is now enlarged, i.e if there are $k$ parts of $(d_i,n_i)$ then $|\mathrm{Aut}(P^*)|$ will include a factor of $k!$. Pick up all these considerations, we have
\begin{equation*} 
\begin{aligned}
G(z,\hbar)^\Delta&=\sum_{n=0}^{\infty}\sum_{d=0}^{\infty}\frac{q^d}{n!} \sum_{P^*\in\mathrm{Part}^*_{d}[n]}\frac{1}{|\mathrm{Aut}(P^*)|}\frac{n!}{n_1!n_2!\cdots n_{l(P)}!}\prod_{i=1}^{l(P^*)}G_{d_i}^{\circ}(z_{P^*_i},\hbar)^\Delta\\
&=\sum_{n=0}^{\infty}\sum_{d=0}^{\infty}\sum_{P^*\in\mathrm{Part}^*_{d}[n]}\frac{1}{|\mathrm{Aut}(P^*)|}\prod_{i=1}^{l(P^*)}\frac{q^{d_i}}{n_i!}G_{d_i}^{\circ}(z_{P^*_i},\hbar)^\Delta
\end{aligned}
\end{equation*}
by the definition \eqref{recursiveG}, which is precisely the rhs of Eq.~(\ref{discon}).\par

\end{proof}

Note this Lemma is stand for both $r=1$ and $r\neq1$.

%%%%%%%%%%%%%%%%%%%%%%%%%%%%%%%%%%%%%%%%%%%%%%%%%%%%%%%%%%%%%%%%%%%%%%%%%%%%%%%%%%%
\subsection{Proof of the Identity \eqref{Xdidentity}}In this appendix we will prove
\begin{lemma} \label{lemA1}
For any $r,d \in\mathbb Z_+$, the following identity holds
\begin{equation*}
X_d :=\frac{1}{d!(rd)!r^{d}}\sum_{|\lambda|=rd}\chi^{\lambda}_{(1)^{rd}}\chi^{\lambda}_{(r)^d} \prod_{i=1}^{\infty}\frac{x+(i-\lambda_i)\hbar}{x+i\hbar}
 = \frac{(-1)^d}{r^dd!}\prod_{i=1}^{rd}\frac{\hbar}{x+i\hbar}:=\tilde{X}_d.
\end{equation*}
\end{lemma}
We define $L(x)$ by
\begin{equation*}
L(x) = \sum_{|\lambda|=rd}\frac{\chi^{\lambda}_{(r)^d}\chi^{\lambda}_{(1)^{rd}}}{(rd)!}\prod_{i=1}^{rd}\left(x+(i-\lambda_i-\frac{1}{2})\hbar\right),
\end{equation*}
then Lemma \ref{lemA1} is equivalent to
\begin{equation}\label{aim1}
\left[x^{k}\right]L(x) = 0,\quad k =1,2,\cdots, rd
\end{equation}
and
\begin{equation*}
\left[x^{0}\right]L(x) = (-1)^d\hbar^{rd}.
\end{equation*}
First we have the following lemma
\begin{lemma}
The following identity holds for arbitrary partitions $\mu,\nu$
\begin{equation}\label{chiortho}
\sum_{\lambda}\chi^{\lambda}_{\mu}\chi^{\lambda}_{\nu} = \delta_{\mu,\nu}.
\end{equation}
\end{lemma}
\begin{proof}
Recall the orthogonality
\begin{equation}\label{ortho}
\braket{\mu|\nu}=\delta_{\mu,\nu},
\end{equation}
by noticing degree d Schur functions consist a complete basis for the space of degree d symmetric polynomials, by Boson-Fermion correspondence, we have $\{\ket{\mu},|\mu| = d\}$ forming a complete basis for the subspace of energy $d$ of the semi-infinite wedge space $\Lambda_0^{\frac{\infty}{2}}V$ (with 0 charge). Therefore, we can rewrite the identity operator  as
\begin{equation*}
\mathrm{Id}_{d} = \sum_{|\lambda|=d}\ket{\lambda}\bra{\lambda},
\end{equation*}
inserting back into Eq.~(\ref{ortho}) will give Eq.~(\ref{chiortho}).
\end{proof}
From the above lemma, one can  we immediately check the first two orders of Eq.~(\ref{aim1}). The first order is just the lemma itself
\begin{equation*}
[x^{rd}]L(x) = \sum_{|\lambda|=rd}\frac{\chi^{\lambda}_{(r)^d}\chi^{\lambda}_{(1)^{rd}}}{(rd)!} = 0.
\end{equation*}
For the second order, we have
\begin{equation*}
[x^{rd-1}]L(x) = \sum_{|\lambda|=rd}\frac{\chi^{\lambda}_{(r)^d}\chi^{\lambda}_{(1)^{rd}}}{(rd)!} \sum_{i=1}^{rd}(i-\lambda_i-\frac{1}{2})\hbar,
\end{equation*}
since 
\begin{equation*}
\sum_{i=1}^{rd} \lambda_i = rd.
\end{equation*}
we have 
\begin{equation*}
\sum_{i=1}^{rd}(i-\lambda_i-\frac{1}{2})\hbar = \left(\frac{(rd+1)rd}{2}-rd-\frac{rd}{2}\right)\hbar,
\end{equation*}
therefore $[x^{rd-1}]L(x)=0$. Now we expand Eq.~(\ref{aim1}) to lower degree of $x$, we have
\begin{equation*}
[x^{rd-k}]L(x) = \sum_{|\lambda|=rd}\frac{\chi^{\lambda}_{(r)^d}\chi^{\lambda}_{(1)^{rd}}}{(rd)!}\mathrm{e}_k((1-\lambda_1-\frac{1}{2})\hbar,\cdots,(rd-\lambda_{rd}-\frac{1}{2})\hbar),
\end{equation*}
where $\mathrm{e}_k$ is the $k$th elementary symmetric polynomial, which is related to Newton polynomial $p_k$, by the following Newton identities
\begin{lemma}
We have the Newton identities for the symmetric polynomials
\begin{equation}\label{relen}
\mathrm{e}_{n} = (-1)^n\sum_{m_1+2m_2+\cdots+nm_n}\prod_{i=1}^n\frac{(-p_i)^{m_i}}{m_i!i^{m_i}},\quad m_i\geq0.
\end{equation}
\end{lemma}
Therefore we have, in order to prove Eq.~(\ref{aim1}), it is sufficient to prove for all $n\leq rd-1$
\begin{equation*}
 D^{d}_\mathbf{k}=\sum_{|\lambda|=rd}\chi^{\lambda}_{(r)^d}\chi^{\lambda}_{(1)^{rd}}\prod_{j=1}^np_{k_j}\left(\lambda_i-i+\frac{1}{2}\right) = 0,\quad \forall \sum_{i=1}^nk_i\leq rd-1,d\geq0,k_i\geq1,
\end{equation*}
where $\mathbf{k}=\{k_1,k_2,\cdots,k_n\}$.  The importance of the requirement $n\leq rd-1$ and further $\sum_{i=1}^nk_i\leq rd-1$, as we will see, will be manifest when we utilizing the commutation relation to contract the insertions, after rewriting these expressions in terms of vacuum expectation value.\par
From Lemma. \ref{proeigen}, we can reformulate the above equation in terms of correlation functions (using the the orthogonality Eq.~(\ref{chiortho}) to drop the constant terms)
\begin{equation}\label{Ddk}
 D^{d}_\mathbf{k} := \braket{0|\alpha_{1}^{rd}W_{0}^{k_1}\cdots W_0^{k_n}\alpha_{-r}^d|0}=0,
\end{equation}
where for $\forall r,s\in\mathbb{Z},\,s\geq0$ the operator $W_{r}^s$ was defined in section 1.1:
\begin{equation*}
W_{r}^{s} = \sum_{k\in\mathbb{Z}+\frac{1}{2}}k^s:\psi_{k-r}\psi_{k}^*:.
\end{equation*}
From this definition one can easily notice $W_r^s$ annihilate the vacuum and covacuum when $r>0$ and $r<0$, respectively.\par
For bosonic generators, we have $\alpha_{n} = W_{n}^0$, and their commutation relation is given by
\begin{proposition}
The commutation relations between $W_{r}^{s}$ and $\alpha_{n}$ is given by
\begin{equation*}
[W_{r}^{s},\alpha_{n}] =\sum_{i=1}^s(-n)^i\left(
\begin{matrix}
s\\
i\\
\end{matrix}
\right)W_{r+n}^{s-i} + c^{r}_{s,n}\cdot\delta_{r,-n}.
\end{equation*}
\end{proposition}

\begin{proof}
First we notice the only non-vanishing commutators between the bifermions, are those containing common pairs of creators and annihilators, therefore we have
\begin{equation*}
[W_{r}^{s},\alpha_{n}] = \sum_{k\in\mathbb{Z}+\frac{1}{2}}k^s\left(\left[:\psi_{k-r}\psi_{k}^*:,:\psi_{k}\psi_{k+n}^*:\right]+\left[:\psi_{k-r}\psi_{k}^*:,:\psi_{k-r-n}\psi_{k-r}^*:\right]\right)
\end{equation*}
by using the definition of normal order and $\{\psi_k,\psi_k^*\}=1$, we have
\begin{equation}\label{commuwa}
\begin{split}
[W_{r}^{s},\alpha_{n}] &= \sum_{k\in\mathbb{Z}+\frac{1}{2}}k^s\left(:\psi_{k-r}\psi_{k+n}^*:-:\psi_{k-r-n}\psi_{k}^*:\right)+c^{r}_{s,n}\cdot\delta_{r,-n}\\
&=\sum_{k\in\mathbb{Z}+\frac{1}{2}}\left((k-n)^s:\psi_{k-r-n}\psi_{k}^*:\right)-W_{r+n}^{s}+c^{r}_{s,n}\cdot\delta_{r,-n}\\
&=\sum_{i=1}^s(-n)^i\left(
\begin{matrix}
s\\
i\\
\end{matrix}
\right)W_{r+n}^{s-i}+c^{r}_{s,n}\cdot\delta_{r,-n}
\end{split}
\end{equation}
\end{proof}

\begin{remark}
$c^{r}_{s,n}$ is a central term (commuting with all other operators), appears when the diagonal term is generated, however whose explicit form does not concern our purposes. All the commutation relation in the sequel will be understood stand up to a central term.
\end{remark}
By the same spirit, we calculate the commutators for general pair of $W_r^s$ and $W_p^q$ (without loosing generality, we assume $s\geq q$):
\begin{proposition}
The commutation relations between $W_{r}^{s}$ and $W_p^q$ with $s\geq q$ is given by
\begin{equation*}
[W_r^s,W_p^q] = \sum_{u= q}^{s-1}a^{s,q}_{r,p}(u)W^{u}_{r+p}
\left(\begin{matrix}
s\\
i\\
\end{matrix}
\right)W_{r+n}^{s-i}.
\end{equation*}
where the coefficients $a^{s,q}_{r,p}(u) = [k^{u}]\left((k-p)^sk^q-k^s(k-r)^q\right)$. 
\end{proposition}
\begin{proof}
\begin{equation}\label{commuw}
\begin{split}
[W_r^s,W_p^q] &=  \sum_{k\in\mathbb{Z}+\frac{1}{2}}k^s(k+p)^q[:\psi_{k-r}\psi_{k}^*:,:\psi_{k}\psi_{k+p}^*:]+ \sum_{k\in\mathbb{Z}+\frac{1}{2}}k^s(k-r)^q[:\psi_{k-r}\psi_{k}^*:,:\psi_{k-r-p}\psi_{k-r}^*:]\\
&=\sum_{k\in\mathbb{Z}+\frac{1}{2}}k^s(k+p)^q:\psi_{k-r}\psi_{k+p}^*:-\sum_{k\in\mathbb{Z}+\frac{1}{2}}k^s(k-r)^q:\psi_{k-r-p}\psi_{k}^*:\\
& = \sum_{k\in\mathbb{Z}+\frac{1}{2}}\left((k-p)^sk^q-k^s(k-r)^q\right):\psi_{k-r-p}\psi_{k}^*::=\sum_{u= q}^{s-1}a^{s,q}_{r,p}(u)W^{u}_{r+p}
\end{split}
\end{equation}
Apparently, by the assumption $s\geq q$, only when $q\leq u<s$ are the coefficients $a^{s,q}_{r,p}(u)$ non-vanishing.
\end{proof}
Recall the definition of energy operator $H=W_0^1$, by using the above commutation relation we have, the $W_r^s$ operators have energy $-r$, i.e.
\begin{equation*}
[H,W_r^s] = -rW_r^s.
\end{equation*}
Now, we will begin to prove Eq.~(\ref{Ddk}), 
\begin{proposition}\label{pd}
The following evaluation of vev holds
\begin{equation*}
 D^{d}_\mathbf{k} := \braket{0|\alpha_{1}^{rd}W_{0}^{k_1}\cdots W_0^{k_n}\alpha_{-r}^d|0}=0,\quad \forall \sum_{i=1}^nk_i\leq rd-1,d\geq0,k_i\geq1,
\end{equation*}
\end{proposition}
\begin{proof}
We defined the level for any sequence of $W$ operators $w=W_{q_1}^{k_1}\cdots W_{q_n}^{k_n}$  as $\mathrm{deg}(w) = \sum_{i=1}^{n}k_i$. Note this level is not well defined when viewed the sequence of $W$ operators as an element in the $W^{1+\infty}$ algebra, therefore we cannot say things like the level of a vacuum expectation value.\par
We continue with considering implementing a series of operations for:
\begin{equation*}
D^{d}_{\mathbf{k}} = \braket{0|\alpha_{1}^{rd}W_{0}^{k_1}\cdots W_{0}^{k_n}\alpha_{-r}^d|0}=0,\quad \forall \sum_{i=1}^nk_i\leq rd-1,d\geq0,k_i\geq1,
\end{equation*}

In the following operations we will regard $\alpha$ operators as level 0 $W$ operators. The operation can be recursively implement as follow:\par
1. Using the commutation relation Eq.~(\ref{commuw}) to commute the first $W$ operators (since after commutation there may be more but finite parts of contributions, we will treat all these contributions by the same procedure), which will be treat as target operator, in the sequence of insertion (i.e. $D = \left<w\right>$) with non-negative energy and positive level to the right (if there is no such operator commute the first $W$ operator with negative energy and positive level to the left). We will regard the product (if the contraction occurs) or itself (if permuting occurs) as the new target $W$ operator.\par

2. If the sign of the target energy is the same as its formal one, we will keep commuting it on the same direction. On the other hand, the commutation will change its direction when the energy of the product after commutation (they all have the same energy by Eq.~(\ref{commuw})) changed its sign (we treat 0 as with $+$ sign). \par

3. The commutation will stop if: 1' If the target $W$ operator with negative (resp. positive) energy ever adjunct to the vaccum (resp. covacuum), by the fact they annihilate the corresponding state, these parts of contribution will equals to 0. 2' If the target $W$ operator has 0 level, i.e. it becomes $\alpha$ operator. If there is $\alpha_0$ operator in the operator sequence, this part of contribution will equals to 0, by the fact we are working in $\Lambda_0^{\frac{\infty}{2}}$

4. Using the commutation relation Eq.~(\ref{commuwa}) to commute all the $\alpha$ operators in each squence of insertion to the leftmost, if they have negative energy ($\alpha_{-n},n>0$), or to the rightmost, if they have positive energy ($\alpha_{n},n\geq0$). Since all the $\alpha$ operators with negative (resp. positive) energy commute with each other, we can sort them both in ascending order w.r.t the energy.\par

5. If there are $W$ operators with positive level remains, them repeat from procedure 1. If there is no $W$ operators with positive level, then stop.\par

After finite times of recursion, the operation will terminate, and we will arrive at:
\begin{equation*}
 D^{d}_\mathbf{k} = \sum_{\mathbf{p},\mathbf{q}}a_{\mathbf{p},\mathbf{q}}\left<0\right|\alpha_{p_1}\alpha_{p_2}\cdots\alpha_{p_{l(\mathbf{p})}} \alpha_{q_1}\alpha_{q_2}\cdots\alpha_{q_{l(\mathbf{q})}} \left|0\right>,
\end{equation*}
where in each part of contribution $p_1\geq p_2\geq\cdots\geq p_{l(\mathbf{p})}>0>q_1\geq q_2\geq\cdots\geq q_{l(\mathbf{q})}$, with $a_{\mathbf{p},\mathbf{q}}$ the corresponding coefficients. We first note the contribution can only be non-vanishing if the $\alpha$ operators coming in pairs, i.e. $p_i=q_{l(\mathbf{q})}-i+1,l(\mathbf{p})=l(\mathbf{q})$. Therefore we are left with
\begin{equation*}
 D^{d}_\mathbf{k} = \sum_{\mathbf{p}}a_{\mathbf{p},\mathbf{q}}\left<0\right|\alpha_{p_1}\alpha_{p_2}\cdots\alpha_{p_{l(\mathbf{p})}} \alpha_{-p_{l(\mathbf{p})}}\alpha_{-p_{l(\mathbf{p})-1}}\cdots\alpha_{-p_1} \left|0\right>
\end{equation*}
Now, from the commutation relations, we observe an important fact: if the number in the sequence of insertion reduce by 1 its level will reduce by at least one. \par

Next, in order to generate a pair of $\alpha_{-k},\alpha_k$, if $k\neq1,r$, one need to use at least 1 positive level $W$ operator as a bridge and at least $r$ $\alpha_1$ operators and 1 $\alpha_{-r}$ operator as source.  The total reduction of level will be at least r+1. However, when $k=1,r$, since one of the corresponding $\alpha$ operator already exist, we can save one time of using the bridge, therefore the total reduction of level will be at least r. We will denote the number of pairs with $k\neq1,-1,r,-r$ as $n_1$ and $k=1,-1,r,-r$ as $n_2$.\par

By the above argument we can get several inequalities: First, since the total number of source is $rd+d$ we have

$$ 
n_1+n_2\leq d.
$$

Second, by our assumption, we have 

$$
\sum_{i=1}^{n}k_i:=s\leq rd-1.
$$

Third, since decrease in number of the insertions must no more than the total level, we have

$$
2(n_1+n_2)\geq rd+d+n-s.
$$

Forth, since each of the pair will cost at least one positive level $W$ operator to form, we have 
$$
n_1+n_2\leq n.
$$

 Aside from these four inequalities, we also have the fifth, which will be used in the prove of next proposition, i.e. the total reduction of level will be no less than those needed for forming the pairs: 
 
 $$
 (r+1)n_1+rn_2\leq s.
 $$\par

Now combining the second and third inequalities will give us
\begin{equation*}
2(n_1+n_2)\geq rd+d+n-s \geq d+n+1,
\end{equation*}
combining this inequality with the forth one, gives
\begin{equation*}
\begin{split}
n_1+n_2&\geq \frac{d+n+1}{2}\geq\frac{n_1+n_2+d+1}{2}\\
\Longrightarrow n_1+n_2&\geq d+1.
\end{split}
\end{equation*}
Which is a contradiction when we taking the first inequality into consideration.\par
Therefore, there can be no non-vanishing contribution after the operation, i.e. $ D^{d}_\mathbf{k} =0$. 
\end{proof}
Since  Eq.~(\ref{aim1}) is equivalent to Proposition \ref{pd}, it is thus proved.\par
Our final aim will be proving:
\begin{equation}\label{aim2}
\left[x^{0}\right]L(x) = \left[x^{0}\right]\sum_{|\lambda|=rd}\frac{\chi^{\lambda}_{(r)^d}\chi^{\lambda}_{(1)^{rd}}}{(rd)!}\prod_{i=1}^{rd}\left(x+(i-\lambda_i-\frac{1}{2})\hbar\right)= (-1)\hbar^{rd}.
\end{equation}
Again, transforming the above equation to the operator formalism, we only need to prove the following proposition
\begin{proposition}\label{laim2}
\begin{equation*}
 D^{d}_\mathbf{k} := \braket{0|\alpha_{1}^{rd}W_{0}^{k_1}\cdots W_0^{k_n}\alpha_{-r}^d|0} = d!(rd)!r^d\cdot\delta_{(r)^d,\mathbf k},
\end{equation*}
for $s:=\sum_{i=1}^{n}k_i=rd$.
\end{proposition}

\begin{proof}
 By the same argument as in the proof of Proposition \ref{pd}, we will arrive at an inequality
\begin{equation*}
d \geq n_1+ n_2\geq\frac{d+n}{2}\geq\frac{d+n_1+n_2}{2}\geq d.
\end{equation*}
%%Fupo
%%
The equality will be satisfied for $n_1+n_2 = n = d$, and recall the fifth inequality stated before, we have now 
\begin{equation*}
rd=r(n_1+n_2)\leq(r+1)n_1+rn_2\leq s=rd,
\end{equation*}
the equality can only be saturated for $n_1=0,n_2=d$. Now since we have $d$ pairs of $\alpha$ which need at least 1 positive level $W$ to form. However, due to $n=d$, each pair can only use 1 $W$, which force it has level no less than $r$, and since $s=rd$, their level should all equals to $r$. Therefore the only non-vanishing contribution will come from:
\begin{equation*}
D^{d}_{(r)^d}=\left<0\right|\alpha_{1}^{rd}W_0^{r}W_0^{r}\cdots W_0^{r}\alpha_{-r}^d\left|0\right>,
\end{equation*}
and having the following form
\begin{equation}\label{resultform}
\left<0\right|\left(\alpha_1\alpha_{-1}\right)^{d_1}\left(\alpha_r\alpha_{-r}\right)^{d_2}\left|0\right>,\quad d_1+d_2=d.
\end{equation}
which equals to $r^{d_2}$ (by the commutation relation $[\alpha_r,\alpha_{-r}]=r$) for given $(d_1,d_2)$.\par
Besides, in order for the equality to hold, we also need the requirement that whenever we contract two $W$ operators, only the leading level contribution will remain. Now we can calculate the above vacuum expectation value by the restricted version of Eq.~(\ref{commuwa}) (we do not need the restricted commutation relation between general $W$ operators, since they do not allowed to commute with each other in this case):
\begin{equation}\label{reslaw}
[W_r^s,\alpha_{n}]^r = -nsW_{r+n}^{s-1}.
\end{equation}
then in order to get the contribution in the form of Eq.~(\ref{resultform}), we first commute the $\alpha_{-r}$ to the right of the $W$ operators and $\alpha_1$ to their left, this can be done since these $\alpha$ operators commute with each other in the case $r\neq1$. We note although all the insertions are ordered, the $\alpha_{-1}$ after commute to the left of $W$ operators are unordered since they all commute each other. Therefore choice of the $\alpha$ operators will give us an prefactor
\begin{equation*}
\left(
\begin{split}
&d\\
1,1&,\cdots,1
\end{split}\right)
\left(
\begin{split}
&rd\\
r,r&,\cdots,r
\end{split}\right)=\frac{d!(rd)!}{(r!)^d}
\end{equation*}
now we have
\begin{equation*}
\begin{split}
D^{d}_{(r)^d}&=\left<0\right|\alpha_{1}^{rd}W_0^{r}W_0^{r}\cdots W_0^{r}\alpha_{-r}^d\left|0\right>\\
&=\frac{d!(rd)!}{(r!)^d}\left<0\right|\left(\alpha_{1}^{r}W_0^{r}\alpha_{-r}\right)^d\left|0\right>
\end{split}
\end{equation*}
Now, we use the restricted commutation relation Eq.~(\ref{reslaw}) to contract (the permute part in the commutation law is fobidden) the $\alpha$ operators with the $W_0^{r}$ operator. If the insertion $\alpha_{1}^{r}W_0^{r}\alpha_{-r}$ results into $\alpha_{r}\alpha_{-r}$ then we will get a prefactor $r!$, and if the result insertion is $\alpha_1\alpha_{-1}$, we will get a prefactor $r\cdot r!$, also, for the order of these pairs, we have a prefactor $\left(\begin{matrix}d\\d_1\\ \end{matrix}\right)$. Then for fixed $(d_1,d_2)$, we have
\begin{equation*}
\begin{split}
D^{d_1,d_2}_{(r)^d}&=\frac{d!(rd)!}{(r!)^d}\left(\begin{matrix}d\\d_1\\ \end{matrix}\right)(r\cdot r!)^{d_1}(r!)^{d_2}\left<0\right|\left(\alpha_1\alpha_{-1}\right)^{d_1}\left(\alpha_r\alpha_{-r}\right)^{d_2}\left|0\right>\\
&=\frac{d!(rd)!}{(r!)^d}\left(\begin{matrix}d\\d_1\\ \end{matrix}\right)(r\cdot r!)^{d_1}(r!)^{d_2}r^{d_2}=d!(rd)!r^d\left(\begin{matrix}d\\d_1\\ \end{matrix}\right).
\end{split}
\end{equation*}
Summing $d_2$ from 0 to $d$ gives us,
\begin{equation*}
D^{d}_{(r)^d} = d!(rd)!r^d
\end{equation*}
where we have divided an overall factor $d!$ for permuting the $W_0^r$ operators, since one cannot know a priori which of the $W_0^r$ will result into pair of  $\alpha_1\alpha_{-1}$ or $\alpha_r\alpha_{-r}$.\par
Now, recall
\begin{equation*}
[x^{0}]L(x) = \sum_{|\lambda|=rd}\frac{\chi^{\lambda}_{(r)^d}\chi^{\lambda}_{(1)^{rd}}}{(rd)!}\mathrm{e}_{rd}((1-\lambda_1-\frac{1}{2})\hbar,\cdots,(rd-\lambda_{rd}-\frac{1}{2})\hbar),
\end{equation*}
by the relationship between $\mathrm{e}_n$ and $p_k$ Eq.~(\ref{relen}), i.e. the Newton identity, we get the coefficients in front of the contribution $p_r^d$ is $\frac{1}{(rd)!}(-1)^{rd}(-1)^d\frac{1}{d!r^d}$, and the factor $(-1)^{rd}$ is cancelled by
\begin{equation*}
\mathrm{e}_{rd}((1-\lambda_1-\frac{1}{2})\hbar,\cdots,(rd-\lambda_{rd}-\frac{1}{2})\hbar) = (-1)^{rd}\mathrm{e}_{rd}((\lambda_1-1+\frac{1}{2})\hbar,\cdots,(\lambda_{rd}-rd+\frac{1}{2})\hbar)
\end{equation*}
therefore we arrive at
\begin{equation*}
[x^{0}]L(x)=\frac{\hbar^{rd}(-1)^d}{d!r^d(rd)!}\left<0\right|\alpha_{1}^{rd}W_0^{r}W_0^{r}\cdots W_0^{r}\alpha_{-r}^d\left|0\right>=(-1)^d\hbar^{rd}.
\end{equation*}
\end{proof}
Combining Proposition \ref{pd} and Proposition \ref{laim2}, we have shown
\begin{equation*}
X_d = \tilde{X}_d,\quad \forall d .
\end{equation*}
This finish the proof of Lemma \ref{lemA1}.\par
%%%%%%%%%%%%%%%%%%% %%%%%%%%%%%%%%%%%%%%%%%%%%%%%%%%%%%%%%%%%%%%%%%%%%%%%%%%%%%%%%%%
\subsection{Proof of Equation \eqref{AAB}}
In this appendix we will prove 

\begin{lemma}
Recall  the Matrices $A^{+-}, A^{--}, B^{+-}$ are defined by % in Eq.~(\ref{Aminus},\ref{Aplus}) are
\begin{equation*}
\begin{split}
A^{--}_{i,j} &= \sum_{d=0}^{\infty}\frac{(-1)^dq^{rd}}{r^dd!\hbar^{(r+1)d}}[x^{i-1}]\rho(x)x_{[rd-j+1]} ,\\
A^{+-}_{-i,j} &= \sum_{rd>j-1}\frac{(-1)^dq^{rd}}{r^dd!\hbar^{(r+1)d}}[x^{-i}]\rho(x)x_{[rd-j+1]},\\
B^{+-}_{-i,j} & = [x^{-i}y^{-j}]\frac{\rho(x)}{\rho(y)}\sum_{d=1}^{\infty}\frac{q^{rd}}{d\cdot r^d\hbar^{(r+1)d}}\sum_{k=0}^{d-1}\frac{(-1)^{k-1}}{k!(d-1-k)!}\sum_{n=1}^{r}x_{[rk+n]}y_{[r(d-k)+1-n]}.
\end{split}
\end{equation*}
%where for $i\leq n+1$
%\begin{equation*}
%e^{i-1}_{n} = \mathrm{e}_{i-1}(0,1,2,\cdots,n-1)
%\end{equation*}
%and $e^{i-1}_{n} := 0 $ for $i>n+1$.
%%
%%
We have the following combinatorial identity
\begin{equation*}
A^{+-}  = B^{+-}A^{--}.
\end{equation*}
\end{lemma}
\begin{proof}
By direct calculation, we have
\begin{equation*}
\begin{split}
(B^{+-}A^{--})_{-i,j} &= \sum_{m=1}^{\infty}B^{+-}_{-i,m}A^{--}_{m,j} = \sum_{m=1}^{\infty}[x^{-i}y^{-m}]\frac{\rho(x)}{\rho(y)}\sum_{d=1}^{\infty}\frac{q^{rd}}{d\cdot r^d\hbar^{(r+1)d}}\\
&\cdot\sum_{k=0}^{d-1}\frac{(-1)^{k-1}}{k!(d-1-k)!}\sum_{n=1}^{r}x_{[rk+n]}y_{[r(d-k)+1-n]} \sum_{d=0}^{[(j-1)/r]}\frac{(-1)^dq^{rd}}{r^dd!\hbar^{(r+1)d}}\left([x^{m-1}]\rho(x)x_{[rd_2-j+1]}\right)\\
&=\sum_{m=1}^{\infty}\sum_{d=1}^{\infty}\frac{q^{rd}}{r^d\hbar^{(r+1)d}}\sum_{d_1=0}^{d}\sum_{k=0}^{d_1-1}\frac{(-1)^{d_2+1-k}}{d_1k!(d_1-1-k)!d_2!}\\
&\cdot\left(\sum_{n=1}^{r}[x^{-i}y^{-m}]\frac{\rho(x)}{\rho(y)}x_{[rk+n]}y_{[r(d_1-k)+1-n]}\right)\left([x^{m-1}]\rho(x)x_{[rd_2-j+1]}\right),
\end{split}
\end{equation*}
which means we need to prove $\forall d\geq0$
\begin{equation}\label{lastl}
\begin{split}
\sum_{m=1}^{\infty}\sum_{d_1=0}^{d}&\sum_{k=0}^{d_1-1}\frac{(-1)^{d_2+d-1-k}d!}{d_1k!(d_1-1-k)!d_2!}\left(\sum_{n=1}^{r}[x^{-i}y^{-m}]\frac{\rho(x)}{\rho(y)}x_{[rk+n]}y_{[r(d_1-k)+1-n]}\right)\\
&\cdot\left([x^{m-1}]\rho(x)x_{[rd_2-j+1]}\right)= [x^{-i}]\rho(x)x_{[rd-j+1]}.
\end{split}
\end{equation}
First we have lhs of Eq.~(\ref{lastl}) can be written as (noticing we have $rd>j-1$)
\begin{equation*}
lhs = \sum_{d_2=0}^{[(j-1)/r]}\sum_{k=0}^{d_1-1}\frac{(-1)^{d_2+d-1-k}d!}{d_1k!(d_1-1-k)!d_2!}\left(\sum_{n=1}^{r}[x^{-i}]\rho(x)x_{[rk+n]}[y^{-1}]\prod^{r(d_1-k)-n}_{p=j-1-rd_2}\frac{\hbar}{y-p\hbar+\frac\hbar 2}\right),\quad d_1+d_2=d,
\end{equation*}
by using
\begin{equation*}
\sum_{m=1}^{\infty}[y^{-m}]\frac{1}{\rho(y)}y_{[r(d_1-k)+1-n]}[x^{m-1}]\rho(x)x_{[rd_2-j+1]}=[y^{-1}]\prod^{r(d_1-k)-n}_{p=j-1-rd_2}\frac{\hbar}{y-p\hbar+\frac\hbar 2},
\end{equation*}
and the fact that the term
\begin{equation*}
[y^{-1}]\prod^{r(d_1-k)-n}_{p=j-1-rd_2}\frac{\hbar}{y-p\hbar+\frac\hbar2},
\end{equation*}
can only be non-vanishing ( and equals to 1) for 
\begin{equation*}
j-1-rd_2 = r(d_1-k)-n\Longleftrightarrow n=-j+1+r(d-k),
\end{equation*}
since $1\leq n\leq r$. This constraint for $n$ also constrain the value of $k$ by:
\begin{equation*}
d-\frac{r-1+j}{r}\leq k\leq d-\frac{j}{r}.
\end{equation*}
Furthermore, by $k\in\mathbb{Z}$ and $d-\frac{j}{r}-\left(d-\frac{r-1+j}{r}\right)=\frac{r-1}{r}<1$, we have $k=\left[d-\frac{j}{r}\right]$. Noticong that one has
\begin{equation*}
-\left[-\frac{j}{r}\right]-1=\left[\frac{j-1}{r}\right],
\end{equation*}
since $d_2\leq\left[\frac{j-1}{r}\right]$, we have $k\leq d_1-1$ coincide with its definition for $\forall d_2$.\par
Therefore we have
\begin{equation*}
\begin{split}
lhs &= \sum_{d_2=0}^{[(j-1)/r]}\frac{(-1)^{d_2+d-1-k}d!}{d_1k!(d_1-1-k)!d_2!}[x^{-i}]\rho(x)x_{[rd+1-j]} \\
&= [x^{-i}]\rho(x)x_{[rd+1-j]}\sum_{d_2=0}^{d-k-1}\frac{(-1)^{d_2+d-1-k}d!}{d_1k!(d-d_2-1-k)!d_2!}.
\end{split}
\end{equation*}
Therefore, if the following combinatorial indentity hold, the proof will be completed.
\begin{equation}\label{cbl}
\sum_{d_2=0}^{d-k-1}\frac{(-1)^{d_2+d-1-k}d!}{d_1k!(d-d_2-1-k)!d_2!}=1,\quad k=\left[d-\frac{j}{r}\right].
\end{equation}
\end{proof}
The only task remains is to prove Eq.~(\ref{cbl}), which is a specialization of the following lemma
\begin{lemma}
$\forall d_1,d_2\in\mathbb{Z}$, denote $d = d_1+d_2$, we have $\forall k\geq0, k\leq d$ 
$$
\sum_{d_2=0}^{d-k-1}\frac{(-1)^{d_2+d-1-k}d!}{d_1k!(d-d_2-1-k)!d_2!}=1.
$$
\end{lemma}
\begin{proof}
Easy to prove by induction on $k$, $\forall d_1,d_2$.
\end{proof}
\end{appendix}

%%%%%%%%%%%%%%%%%%%%%%%%%%%%%%%%%%%%%%%%%%%%%%%
\renewcommand{\refname}{Reference}
\bibliographystyle{plain}
\bibliography{qcbib.bib} 
%%%%%%%%%%%%%%%%%%%%%%%%%%%%%%%%%%%%%%%%%%%%%%%%%%%%%%%%%%%%%%%%%%%%%%%%%%%%%%%%%%%
\end{document}